\documentclass[12pt,centertags,reqno]{amsart}

\usepackage{latexsym}
\usepackage[english]{babel}
\usepackage[T1]{fontenc}
\usepackage[numbers]{natbib}
\usepackage{amssymb}
\usepackage{fancyhdr}
\usepackage{url}
\usepackage{hyperref}
\usepackage{verbatim}
\usepackage{leftidx}
\usepackage{epsfig}
 \usepackage{graphicx}
 \usepackage{epstopdf}

\usepackage{mathrsfs, mathtools}
\usepackage{stmaryrd}

\usepackage{marvosym}
\mathtoolsset{showonlyrefs}

\newtheorem{thm}{Theorem}[section]

\newtheorem{rem}[thm]{Remark}

\newtheorem{definition}[thm]{Definition}
\newtheorem{theorem}{Theorem}[section]
\newtheorem{proposition}[thm]{Proposition}

\numberwithin{equation}{section}

\def \F {\mathcal F}
\def \G {\mathcal G}

\def \L {\mathcal L}

\def \P {\mathbf P}

\def \I {{\mathbf 1}}

\def \R {\mathbb R}
\def \bF {\mathbb F}
\def \bG {\mathbb G}

\def \bN {\mathbb N}

\def \bpi {\boldsymbol{\pi}}

\newcommand{\ho}{h^{(1)}}
\newcommand{\hd}{h^{(2)}}
\newcommand{\hm}{h^{(m)}}

\newcommand{\balpha}{\boldsymbol{\alpha}}
\newcommand{\bmu}{\boldsymbol{\mu}}

\newcommand{\ud}{\mathrm d}

\newcommand{\esp}[2][\mathbb E] {#1\left[#2\right]}

\newcommand{\p}{\mathbf p}

\begin{document}

\title{Optimal Convergence Trading with Unobservable Pricing Errors}
%\thanks{}

\author[S.~ALTAY]{S\"{u}han Altay}
\address[S.~ALTAY]{Department of Financial and Actuarial Mathematics, TU Wien, Wiedner Hauptstrasse 8-10, 1040 Vienna, Austria.}
\email{altay@fam.tuwien.ac.at}
\author[K.~COLANERI]{Katia Colaneri}
\address[K.~COLANERI]{School of Mathematics, University of Leeds,  LS2 9JT Leeds, UK.}
\email{k.colaneri@leeds.ac.uk}
\author[Z.~EKSI]{Zehra Eksi}
\address[Z.~EKSI]{Institute for Statistics and Mathematics, WU-University of Economics and Business, Welthandelsplatz 1, 1020, Vienna, Austria.}
\email{zehra.eksi@wu.ac.at}

\date{}


\begin{abstract}
We study a dynamic portfolio optimization problem related to convergence trading, which is an investment strategy that exploits temporary mispricing by simultaneously buying relatively underpriced assets and selling short relatively overpriced ones with the expectation that their prices converge in the future. We build on the model of \citet{liu2013optimal} and extend it by incorporating unobservable Markov-modulated pricing errors into the price dynamics of two co-integrated assets. We characterize the optimal portfolio strategies in full and partial information settings both under the assumption of unrestricted and beta-neutral strategies. By using the innovations approach, we provide the filtering equation that is essential for solving the optimization problem under partial information.
Finally, in order to illustrate the model capabilities, we provide an example with a two-state Markov chain.
\end{abstract}

\maketitle \markboth{}{}
\renewcommand{\theequation}{\arabic{section}.\arabic{equation}}
\pagenumbering{arabic}

\section{Introduction}
Convergence-type trading strategies have become one of the most popular trading strategies that are used to capitalize on market inefficiencies, or deviations from ``equilibrium,'' especially with the rapid developments in algorithmic and high-frequency trading. For a typical convergence trade, temporary mispricing is exploited by simultaneously buying relatively underpriced assets and selling short relatively overpriced assets in anticipation that at some future date their prices will have become closer. Thus one can profit by the extent of the convergence. A prime example of a convergence trade is a pairs trading strategy that involves a long position and a short position in a pair of similar stocks that have moved together historically and hence an investor can profit from the relative value trade arising from the cointegration between asset price dynamics involved in the trade. Other examples of convergence-type trading strategies can be given as risk arbitrage (known also as merger arbitrage) that speculates on successful completion of a merger of two companies or cash and carry trade that tries to benefit from pricing inefficiencies between spot market and futures market of the same underlying stock or commodity by simultaneously placing opposite bets on spot and futures markets.

In this work, we extend the convergence trade model given by \citet{liu2013optimal} that investigate the dynamic optimal portfolio allocation via expected utility maximization from terminal wealth, with two co-integrated assets with pricing errors and a market index. \citet{liu2013optimal} show that under recurring and non-recurring ``arbitrage'' opportunities, optimal portfolio allocations could deviate from conventional long-short delta-neutral strategies and it can be optimal to hold both risky assets long (or short) at the same time. We extend \citet{liu2013optimal} mainly in two directions. First, we assume that pricing errors related to the co-integrated assets are modulated by a continuous-time, finite-state Markov chain, which is taken to be unobservable and hence needs to be filtered out. Taking pricing errors (or ``alphas'' as commonly referred in finance literature) dependent on a hidden Markov chain captures certain salient features of convergence trade. Although most of the existing literature assumes that pricing errors are fully observable, in reality, those errors, albeit having a stochastic nature, cannot be known precisely or may depend on some unobservable state variables that change accordingly to certain factors in the economy or the market. By modeling those pricing errors depending on an unobservable regime-switching factor, we would like to build a more realistic representation for convergence trading. The second extension of our model is that we allow capital asset pricing model (CAPM) betas of two risky assets to be different. That enables us to show the optimal portfolio allocation for beta-neutral pairs trading, which is designed to keep the portfolio's beta zero all the time and hence achieve market neutrality. It is a common market practice among pairs traders to have a beta-neutral portfolio to avoid market risk. Moreover, allowing different betas also enables us to account for ``betting against beta'' strategies that involve going short with high-beta stocks and going long with low-beta ones. Betting against beta type strategies are often associated with fluctuating low alpha \citet{frazzini2014betting}, which also justifies our choice of modeling pricing errors under regime switching and partial information.

There is a growing stream of literature about optimal convergence trading. \citet{liu2003losing} has a partial equilibrium examination of convergence trading strategies, where the mispricing is modeled using a Brownian bridge. \citet{jurek2007dynamic} incorporate an Ornstein--Uhlenbeck process to model the spread for non-myopic investors and solve the dynamic portfolio allocation for constant relative risk aversion and recursive Epstein--Zin utility function. By building on the results of \citet{jurek2007dynamic}, \citet{liu2013optimal} solve a similar problem by focusing both on recurring and non-recurring arbitrage opportunities in a continuous error-correction model with two co-integrated assets and a market index. \citet{lei2015costly} extend \citet{liu2013optimal} by incorporating transaction costs. Inspired by the dynamic pairs trading model of \citet{mudchanatongsuk2008optimal}, \citet{tourin2013dynamic} develop an optimal portfolio strategy to invest in two risky assets and the money market account, assuming that log-prices are co-integrated, and solve the optimal portfolio allocation problem for the exponential utility. \citet{cartea2016algorithmic} extend \citet{tourin2013dynamic} to allow the investor to trade in multiple co-integrated assets. \citet{chiu2011mean} investigates the optimal dynamic trading of cointegrated assets using the classical mean-variance portfolio selection criterion. \citet{angoshtari2016market} studies the necessary and sufficient conditions for well-posedness and no-arbitrage for the model of \citet{liu2013optimal} by focusing on the concept of investor nirvana.

Considering similar problems under regime-switching and/or partial information in the literature, studies focusing on the dynamic portfolio choice problem are rather limited. \citet{lee2016pairs} solve the optimal pairs trading problem within a power utility setting, where the drift uncertainty is modeled by a continuous Gaussian mean-reverting process and necessitates Kalman filtering to extract the unobservable state process. \citet{altay2018pairs} extend the pairs trading model of \citet{mudchanatongsuk2008optimal} by incorporating regime switching under partial information and risk penalization. However, classical portfolio selection problems, which do not cover portfolios involving co-integrated assets, in a full or partial information and/or Markov regime-switching framework  can be found, for example, in \citet{zhou2003markowitz}, \citet{bauerle2004port}, and \citet{sotomayor2009explicit} for the full information case with Markov regime switching or \citet{rieder2005portfolio}, \citet{bjork2010optimal} and \citet{frey2012portfolio} for the partial information case.

In summary, we have the following key contributions. First, we compute the optimal unrestricted and beta-neutral strategies both in full and partial information settings for a log-utility trader by using dynamic programming. Second, we characterize the value function as the unique (classical) solution of the Hamilton--Jacobi--Bellman (HJB) equation, which is reduced to a system of ordinary differential equations (ODE) in the full information case, and given by a system of partial differential equations (PDE) in the partial information case. We also provide verification results for both cases. Third, to solve the convergence trade problem under partial information we compute the filtering equation by applying the innovations approach. Having the filter dynamics enables us to  study the equivalent reduced problem where unobservable states of the Markov chain are replaced by their optional projection over the available filtration. Comparing optimal strategies under full and partial information, we obtain that the certainty equivalence principle holds, i.e., the optimal portfolio strategy in the latter case can be obtained by replacing the unobservable state variable with its filtered estimate.  Finally, we analyze  an example with a two-state Markov chain numerically and demonstrate certain features of our proposed model. In particular, we illustrate the dominance of unrestricted strategies over beta-neutral strategies. Moreover, we show that a trader who uses averaged data (in terms of parameters) is not performing better than the trader who uses a Markov modulated model in a full information setting. For the partial information case, our example suggests that there is a non-negative information premium, indicating that the fully informed trader has an advantage over the partially informed one.

The remainder of the paper is organized as follows. Section \ref{sec:setting} introduces the model setting and the notation. In Section \ref{sec:regime_full} we study the portfolio optimization problem in a full information setting with regime switching. In Section \ref{sec:partial_info} we solve the utility maximization problem under partial information. In Section \ref{sec:2stateMC}, we provide a numerical analysis of an example with a two-state Markov chain. We conclude with Section \ref{sec:concl}. In order to improve the flow of the paper we provide proofs of all results in the Appendix.

\section{Model Setting and Notation}\label{sec:setting}
We study a modification of the continuous-time error-correction model of \citet{liu2013optimal} in a regime-switching setup under both full and partial information. Precisely we fix a probability space $(\Omega, \G, \P)$ and a finite time horizon $T$ which coincides with the terminal time of an investment. We also introduce a complete and right-continuous filtration $\bG=\{\G_t, \ t \in [0,T]\}$, representing the full information flow, and assume that all processes defined below are adapted to $\bG$.

Let $Y$ be a continuous-time finite-state Markov chain taking values in $\mathcal{E}=\{e_1, \dots, e_K\}$, for $K\geq 2$ , where, without loss of generality, we assume that $e_i$ is the $i$-th canonical vector in $\R^K$, for every $i\in\{1, \dots, K\}$. We denote by $Q=(q^{ij})_{i,j\in\{1, \dots, K\}}$ the infinitesimal generator of $Y$, with $q^{ij}>0$  for every $i\neq j$ and $q^{ii}=-\sum_{j \neq i}q^{ij}$, and let   $\Pi=(\Pi^1, \dots \Pi^K)$ be its initial distribution.

\begin{rem}\label{Rem1}
The finite-state nature of the Markov chain implies that for every $t\in[0,T]$, and every function $f:\mathcal{E}\to \R$ we have $f(Y_t)=\sum_{i=1}^K f^i \I_{\{Y_t=e_i\}}$ where  $f^{i}=f(e_i)$, for every $i\in \{1, \dots, K\}$.
%. In the sequel, for the ease of notation, we write $f(Y_t)={\bf f} Y_t$ where ${\bf f}=(f^{1}, \dots, f^{K})^\top$.
\end{rem}

We consider a market model where a trader can invest in a riskless asset with constant rate of return $r\ge 0$ and three risky assets with price processes $S^{(m)}, S^{(1)}$ and $S^{(2)}$, where the first one represents the market index and the other two are co-integrated assets. We assume that the price dynamics of market index is given by
\begin{equation}
\frac{\ud S^{(m)}_t}{S^{(m)}_t}=(r+\mu_m)\,\ud t+\sigma_m\,\ud B^{(m)}_t, \quad S^{(m)}>0,
\end{equation}
where $\mu_m\in \R$ is the market risk premium, $\sigma_m>0$ is the market volatility and $B^{(m)}$ is a standard Brownian motion. Moreover co-integrated asset prices are described by the following SDEs,
\begin{align}\label{convergence}
\frac{\ud S^{(1)}_t}{S^{(1)}_t}\!=(r+\beta_1 \mu_m)\,\ud t\!+\!\beta_1\sigma_m\,\ud B^{(m)}_t\!\!\!+\!\sigma\,\ud B^{(0)}_t\!\!\!+b_1\,\ud B^{(1)}_t\!\!\!-\!\lambda_1(Y_t)\!\left(X_t\!-\alpha_1(Y_t)\right)\ud t,
\end{align}
\begin{align}
\label{convergence2}
\frac{\ud S^{(2)}_t}{S^{(2)}_t}\!=(r+\beta_2 \mu_m)\,\ud t\!+\!\beta_2\sigma_m\,\ud B^{(m)}_t\!\!\!+\!\sigma\,\ud B^{(0)}_t\!\!\!+\!b_2\,\ud B^{(2)}_t\!\!\!+\!\lambda_2(Y_t)\!\left(X_t\!-\alpha_2(Y_t)\right)\ud t,
\end{align}
with $S^{(1)}_0>0$ and $S^{(2)}_0>0$. Coefficients $\beta_1 \in \mathbb{R}, \beta_2 \in \mathbb{R}, b_1>0, b_2>0$ and $\sigma>0$ are constant parameters and $(B^{(0)}, B^{(1)}, B^{(2)})$ is a three-dimensional standard Brownian motion independent of $B^{(m)}$.

We define the spread process $X$ by $X_t=\log{S^{(1)}_t}-\log{S^{(2)}_t}$, for every $t \in [0,T]$. Here $X$  represents the mean-reverting component of pricing errors. We assume that $\lambda_1(Y_t)+\lambda_2(Y_t)>0$ $\P$-a.s. for every $t \in [0,T]$,  so that $X$ becomes a stationary process with the dynamics
\begin{align}
\nonumber \ud X_t=&\Big(\Gamma_1-\lambda_1(Y_t)(X_t-\alpha_1(Y_t))-\lambda_2(Y_t)(X_t-\alpha_2(Y_t))\Big{)}\,\ud t \\
&+\left(\beta_1-\beta_2\right)\sigma_m\,\ud B^{(m)}_t+b_1 \ud B^{(1)}_t-b_2\,\ud B^{(2)}_t,\quad X_0\in \R, \label{eq:spread}
\end{align}
where
\begin{align}\label{eq:Gamma1}
\Gamma_1=\left(\beta_1-\beta_2\right)\mu_m-\frac{1}{2}\left( (\beta_1^2-\beta_2^2)\sigma_m^2+b_1^2-b_2^2 \right).\end{align}

In \eqref{convergence} and \eqref{convergence2}, the infinitesimal expected returns are
$$(r+\beta_1 \mu_m)\,\ud t-\lambda_1(Y_t)\!\left(X_t\!-\alpha_1(Y_t)\right)\ud t$$ and
$$(r+\beta_2 \mu_m)\,\ud t+\lambda_2(Y_t)\!\left(X_t\!-\alpha_2(Y_t)\right)\ud t,$$ respectively. It is evident from the form of infinitesimal returns that if $\lambda_j(\cdot)$ is chosen to be identical to zero or $X_t$ is equal to $\alpha_j(\cdot)$, for every $j\in\{1,2\}$, asset price dynamics satisfy the CAPM relation, meaning that CAPM establishes the expected returns correctly and there is no mispricing in either asset. On the other hand, for example, if $-\lambda_1(Y)(X-\alpha_1(Y))>0$, the first asset has a higher expected return than it is justified by its exposure to market risk, and hence has a positive alpha, meaning it is undervalued. By choosing $\lambda_1$ and $\lambda_2$ dependent on the hidden Markov chain $Y$, we therefore, postulate that pricing errors (or alphas) depend on some common factor in the economy or the market that can not be directly observed by the trader. Also, as it is suggested by \citet{liu2013optimal}, one can interpret those pricing errors as reflecting momentarily positive or negative liquidity shocks, which may vanish in liquid markets. For example, because of liquidity effects, stocks listed in S\&P 500 have overstated betas \citet{vijh1994s}, which in turn affects  pricing errors. By assuming a regime-switching framework for pricing errors, we are also able to model such liquidity effects. This type of liquidity effects and related mispricing is actually very important for trades involving dual-listed companies (or so-called ``Siamese twin'' companies), which are incorporated in different countries and listed in different exchanges simultaneously while operating as a single entity. For such companies, since shares listed in different exchanges have same control rights and dividends are based on the same cash flow, most of the mispricing between two stocks is due to liquidity effects arising from stock exchanges that individual shares are traded in and hence prone to different regimes; see \citet{de2009risk} for more on stock price differentials of dual-listed companies.
We should also remark that when we take $\beta_1=\beta_2$ and $b_1=b_2$, the model becomes a regime-switching version of the original one suggested by \citet{liu2013optimal} that involves two assets with the same payoff trading at different prices.

\section{Optimal Convergence Trade under Regime Switching}\label{sec:regime_full}
Let $W^{h}$ be the value of a portfolio $h=(\hm, \ho, \hd)$, where quantities $\hm_t, \ho_t$ and $\hd_t$ denote fractions of the wealth invested  at any time $t \in [0,T]$ in the market index $S^{(m)}$ and in the co-integrated assets with prices $S^{(1)}$ and $S^{(2)}$, respectively. Consequently the percentage of wealth invested in the riskless asset is $1-\hm-\ho-\hd$. We introduce now the suitable set of strategies.

\begin{definition}
A $\bG$-admissible portfolio strategy is a  self-financing, $\bG$-predictable strategy $h=(\hm, \ho, \hd)$ such that
\begin{align}
\esp{\int_0^T \left({\hm_t}^2+{\ho_t}^2+{\hd_t}^2\right)\ud t}<\infty. \label{eq:integrability}
\end{align}
The set of $\bG$-admissible strategies is denoted by $\mathcal{A}$.
\end{definition}
For every  $h=(\hm, \ho,\hd) \in \mathcal{A}$, the dynamics of the convergence trading portfolio is given by
\begin{eqnarray}
\nonumber\frac{\ud W^h_t}{W^h_t}&=&\left(r+\left(\hm_t+\ho_t\beta_1+\hd_t\beta_2 \right)\mu_m+\hd_t\lambda_2(Y_t)(X_t-\alpha_2(Y_t))\right.\\
\nonumber&-&\left.\ho_t\lambda_1(Y_t)(X_t-\alpha_1(Y_t))  \right)\ud t
+\sigma_m\left(\hm_t+\ho_t\beta_1+\hd_t\beta_2\right)\ud B^{(m)}_t\\
\nonumber&+& \sigma\left(\ho_t+\hd_t \right)\ud B^{(0)}_t
+b_1\ho_t \ud B^{(1)}_t+b_2\hd_t \ud B^{(2)}_t,\quad W_0^h>0.
\label{wealthdyn}\end{eqnarray}
We consider a trader with logarithmic preferences and who aims to maximize the expected utility from terminal wealth at time $T$ in a market with regime switching. In this section, we assume that  the trader may directly observe the state of the Markov chain Y that influences the dynamics of price processes and the spread. Formally, we address the following problem
\begin{equation}\mbox{ Maximize } \ \mathbf{E}^{t,w,x,i}[\log W_T^h] \ \mbox{ over all } \  h \in \mathcal{A},\label{eq:objective}\end{equation}
where $\mathbf{E}^{t,w,x,i}$ denotes the conditional expectation given $W_t=w$, $X_t=x$ and $Y_t=e_i$.
We define the value function corresponding to problem \eqref{eq:objective} as
\begin{equation}\label{vf}V(t,w,x,i):=\underset{h\in\mathcal{A}}\sup\,\,\mathbf{E}^{t,w,x,i}\left[\log{W_T^h}\right].\end{equation}
Notice that for a given $h\in\mathcal{A}$, $W^h$ is a controlled process. For the sake of notational simplicity from now on we suppress $h$ dependency and write $W$ instead of $W^h$.

In Theorem \ref{optimalfull}, we address the optimization problem by applying dynamic programming. Our goal is to identify the optimal strategy as well as to characterize the value function as the unique solution of the corresponding HJB equation. This approach permits to examine the value function of the control problem in detail. One could alternatively derive the stochastic representation of the value function and characterize it up to the solution of a system of partial differential equations via Feynman--Kac type arguments for Markov-modulated diffusion processes; see, e.g., \citet{Baran2013} and \citet{escobar2015portfolio}.

In the sequel, we will use the following notation for the partial derivatives: for every function $g:[0,T]\times\mathbf{R}_{+}\times\mathbf{R}\to\mathbf{R},$ we write, for instance, $\frac{\partial g}{\partial t}=g_t$. Also, by Remark~\ref{Rem1} we have that $\lambda_j(e_i)=\lambda_j^i$ and $\alpha_j(e_i)=\alpha_j^i$, for $j\in\{1,2\}$ and $i\in\{1,\dots,K\}$.
\begin{theorem}\label{optimalfull} Consider a trader endowed with a logarithmic utility function. Then the optimal portfolio strategy $h^*=({\ho}^{\ast},{\hd}^{\ast},{\hm}^{\ast})\in \mathcal{A}$ is
\begin{eqnarray}\label{opt1}{\ho}^{\ast}(t,x,i)&=& -\frac{\lambda_1^{i}(x-\alpha_1^i)+\lambda_2^{i}(x-\alpha_2^i)\varrho_2}{b_1^2+b_2^2\varrho_2},\\
\label{opt2}{\hd}^{\ast}(t,x,i)&=&\frac{\lambda_2^{i}(x-\alpha_2^i)+\lambda_1^{i}(x-\alpha_1^i)\varrho_1}{b_2^2+b_1^2\varrho_1},\\
\label{opt3}{\hm}^{\ast}(t,x,i)&=& \frac{\mu_m}{\sigma_m^2}-\beta_1{\ho}^{\ast}(t,x,i)-\beta_2{\hd}^{\ast}(t,x,i),
\end{eqnarray}
with $\varrho_1=\frac{\sigma^2}{ \sigma^2+b_1^2}$ and $\varrho_2=\frac{\sigma^2}{ \sigma^2+b_2^2}$. The value function is of the form
\begin{equation}V(t, w, x, i )=\log(w)+m(t,i)x^2+n(t,i)x+u(t,i), \label{eq:value_function}\end{equation}
where functions $m(t,i)$, $n(t,i)$ and $u(t,i)$ for $i\in\{1, \dots, K\}$ solve the following system of ordinary differential equations
\begin{align}
&m_t(t,i)-2(\lambda_1^i+\lambda_2^i)m(t,i)+\sum_{j=1}^K m(t,j) q^{ij}+\Theta_1^i =0,\label{eq:system1}\\
&n_t(t,i)-(\lambda_1^i+\lambda_2^i)n(t,i)+\sum_{j=1}^K n(t,j) q^{ij}+2(\Gamma_1+\lambda_1^i\alpha_1^i+\lambda_2^i\alpha_2^i)m(t,i)-\Theta_2^i=0,\label{eq:system2}\\
&u_t(t,i)+ \sum_{j=1}^K u(t,j) q^{ij}+\Gamma_2m(t,i)+(\Gamma_1+\lambda_1^i\alpha_1^i+\lambda_2^i\alpha_2^i) n(t,i)+\Theta_3^i=0,    \quad{}\label{eq:system3}
\end{align}
with terminal conditions $m(T,i)=0$, $n(T,i)=0$ and $u(T,i)=0$ for all $i\in \{1,\dots,K \}$, and where $\Gamma_1$ is given in \eqref{eq:Gamma1} and
\begin{align*}
\Theta_1^i&=\frac{b_1^2(\lambda_2^i)^2 + b_2^2(\lambda_1^i)^2 +\sigma^2\left(\lambda_1^i+\lambda_2^i\right)^2}{2\left(b_1^2b_2^2 +\sigma^2(b_1^2+b_2^2)\right)},\\
\Theta_2^i&= \frac{\alpha_1^i\lambda_1^i(\lambda_1^i(b_2^2+\sigma^2)+\lambda_2^i\sigma^2)+\alpha_2^i\lambda_2^i(\lambda_2^i(b_1^2+\sigma^2)+\lambda_1^i\sigma^2)}{b_1^2b_2^2 +\sigma^2(b_1^2+b_2^2)},\\
\Theta_3^i&=\frac{(\alpha_1^i\lambda_1^ib_2)^2+(\alpha_2^i\lambda_2^ib_1)^2+\sigma^2(\alpha_1^i\lambda_1^i+\alpha_2^i\lambda_2^i)^2}{2\left(b_1^2b_2^2 +\sigma^2(b_1^2+b_2^2)\right)}+r+\frac{\mu_m^2}{2\sigma_m^2},\\
%\Gamma_1&=\mu_m(\beta_1-\beta_2)-\frac{1}{2}\left(\sigma_m^2(\beta_1^2-\beta_2^2)+b_1^2-b_2^2 \right), \\
\Gamma_2&=\sigma_m^2(\beta_1-\beta_2)^2+b_1^2+b_2^2.
\end{align*}

\end{theorem}
We provide the proof of the Theorem~\ref{optimalfull} in the Appendix.

\begin{rem}[Discussion on the optimal trading strategy.]\label{remarkfull}
The optimal trading strategy is Markov modulated and has a typical structure of a mean-variance portfolio weights. More specifically, numerator of each portfolio weight $h^{(j)^{\ast}}$, $j\in\{1,2\}$, is associated with regime-switching parameters, $\lambda_1(Y),\lambda_2(Y)$ and $\alpha_1(Y),\alpha_2(Y)$, related to the co-integration between $S^{(1)}$ and $S^{(2)}$, or equivalently, to pricing errors. The denominator, on the other hand, is akin to the idiosyncratic risk components, $b_1$, $b_2$ and $\sigma$. We should also emphasize that  ${\ho}^{\ast}$ and ${\hd}^{\ast}$ do not depend on market parameters, $\beta_1$, $\beta_2$, $\mu_m$ and $\sigma_m$, since the market exposure of each asset is covered by investing in the market index. $\varrho_1$ and $\varrho_2$ can be seen as the \emph{relative idiosyncratic variation} of $S^{(1)}$ (resp. $S^{(2)}$ ) with respect to $S^{(2)}$ (resp. $S^{(1)}$). The role of $\varrho_1$ is actually to scale the contribution of pricing error and the independent idiosyncratic variance of $S^{(2)}$ in ${\ho}^{\ast}$. Naturally, $\varrho_2$ has the analogous interpretation. Note that when $\sigma=0$, meaning that there is no correlation between $S^{(1)}$ and $S^{(2)}$, those contributions vanish and the optimal portfolio weights in each stock only depend on their own pricing errors and idiosyncratic risks. The structure of the market portfolio weight is similar to that in \citet{liu2013optimal} and given by the sum of Sharpe's ratio of the market index and a linear combination of ${\ho}^{\ast}$ and ${\hd}^{\ast}$, weighted by their corresponding betas. Finally as ${\ho}^{\ast}$ and ${\hd}^{\ast}$ are prone to different regimes so is ${\hm}^{\ast}$, even if the market index dynamics is independent of the Markov chain.
\end{rem}

\subsection{Optimal Beta-Neutral Investment}
To achieve market neutrality, traders may chose investment strategies so that the resulting portfolio has zero (CAPM) beta. The goal of this section is to characterize this type of trading strategies which are called {\em beta-neutral}. We should also remind the reader that this type of strategies can also be used for ``betting against betas'' type strategies in which high beta asset (short leg) is deleveraged so that its betas has been decreased to $1$ and the low beta asset (long leg) is leveraged so that its beta has become 1. We start with a formal definition.

\begin{definition} A $\bG$-admissible beta-neutral portfolio strategy is $\bG$-predictable self-financing strategy $h^{\beta}=(h^{(\beta,1)},h^{(\beta,2)},h^{(\beta,m)})$ such that
\begin{equation*}\label{beta-admissfull}
\beta_1h^{(\beta,1)}_t+\beta_2h^{(\beta,2)}_t=0, \quad t \in [0,T],
\end{equation*}
and satisfying $\esp{\int_0^T\left( {h^{(\beta,m)}_t}^2+{h^{(\beta,1)}_t}^2\right)\ud t}<\infty$.
We denote the set of $\bG$-admissible beta-neutral strategies by $\mathcal{A}^{\beta}$.
\end{definition}
In the next theorem we compute the optimal beta-neutral investment strategies and the corresponding value function. The proof of this result replicates that of Theorem \ref{optimalfull} and it is therefore omitted.
\begin{theorem}\label{optimalfullbeta}
Consider a trader with a logarithmic utility function. Then the optimal beta-neutral investment strategy ${h^{\beta}}^{\ast}=({h^{(\beta,1)}}^{\ast},{h^{(\beta,2)}}^{\ast},{h^{(\beta,m)}}^{\ast}) \in \mathcal{A}^{\beta}$ is given by
\begin{align}
&{h^{(\beta,1)}}^{\ast}(t,x,i)= -\frac{\lambda_1^i(x-\alpha_1^i)+\frac{\beta_1}{\beta_2}\lambda_2^i(x-\alpha_2^i)}{b_1^2+\frac{\beta_1^2}{\beta_2^2}b_2^2+\sigma^2\left(1-\frac{\beta_1}{\beta_2}\right)^2}, \\
&{h^{(\beta,2)}}^{\ast}(t,x,i)=-\frac{\beta_1}{\beta_2}{h^{(\beta,1)}}^{\ast}(t,x,i),\\
&{h^{(\beta,m)}}^{\ast}(t,x,i)=\frac{\mu_m}{\sigma_m^2}.
\end{align}
The value function is of the form \begin{equation}V(t, w, x, i )=\log(w)+m(t,i)x^2+n(t,i)x+u(t,i),\end{equation}
where the functions $m(t,i)$, $n(t,i)$ and $u(t,i)$ for $i\in\{1, \dots, K\}$ solve the following system of ordinary differential equations
\begin{align}
&m_t(t,i)-2(\lambda_1^i+\lambda_2^i)m(t,i)+\sum_{j=1}^K m(t,j) q^{ij}+\Phi_1^i =0,\label{eq:system1b}\\
&n_t(t,i)-(\lambda_1^i+\lambda_2^i)n(t,i)+\sum_{j=1}^K n(t,j) q^{ij}+2(\Gamma_1+\lambda_1^i\alpha_1^i+\lambda_2^i\alpha_2^i)m(t,i)-\Phi_2^i=0,\label{eq:system2b}\\
&u_t(t,i)+ \sum_{j=1}^K u(t,j) q^{ij}+\Gamma_2m(t,i)+(\Gamma_1+\lambda_1^i\alpha_1^i+\lambda_2^i\alpha_2^i)n(t,i)+\Phi_3^i=0,    \quad{}\label{eq:system3b}
\end{align}
with terminal conditions $m(T,i)=0$, $n(T,i)=0$ and $u(T,i)=0$ for all $i\in \{1,\dots,K \}$, and where $\Gamma_1$ and $\Gamma_2$ are as given in Theorem \ref{optimalfull} and
\begin{align*}
\Phi_1^i&=\frac{\left(\beta_2\lambda_1^i+\beta_1\lambda_2^i\right)^2}{2\left(b_1^2\beta_2^2+b_2^2\beta_1^2+\sigma^2\left(\beta_1-\beta_2\right)^2\right)},\\
\Phi_2^i&=\frac{\left(\alpha_1^i\beta_2\lambda_1^i+\alpha_2^i\beta_1\lambda_2^i\right)\left(\beta_1\lambda_2^i+\beta_2\lambda_1^i\right)}{b_1^2\beta_2^2+b_2^2\beta_1^2+\sigma^2\left(\beta_1-\beta_2\right)^2},\\
\Phi_3^i&=\frac{\left(\alpha_1^i\beta_2\lambda_1^i+\alpha_2^i\beta_1\lambda_2^i\right)^2}{2\left(b_1^2\beta_2^2+b_2^2\beta_1^2+\sigma^2\left(\beta_1-\beta_2\right)^2\right)} +r+\frac{\mu_m^2}{2\sigma_m^2}.
\end{align*}

\end{theorem}

\begin{rem} Notice that the ratio $\beta_1/\beta_2$ plays the role of $\varrho_2$ in Theorem~\ref{optimalfull}. In addition, setting $\beta_1=\beta_2$ in the current context corresponds to the so-called {\em  delta-neutral strategies}. This is a class of investment strategies that satisfy $h^{(\delta,1)}=-h^{(\delta,2)}$. In other terms, this amounts to invest the same capital in each of the co-integrated stocks. In this setting the optimal delta-neutral strategy is given by
\begin{gather*}
{h^{(\delta,1)}}^{\ast}(t,x,i)=-{h^{(\delta,2)}}^{\ast}(t,x,i)=-\frac{\lambda_1^i(x-\alpha_1^i)+\lambda_2^i(x-\alpha_2^i)}{b_1^2+b_2^2},\\ {h^{(\delta,m)}}^{\ast}(t,x,i)=\frac{\mu_m}{\sigma_m^2}.
\end{gather*}
\end{rem}

\section{Optimal Convergence Trade under Partial Information }\label{sec:partial_info}
The goal of this section is to study the utility maximization problem related to convergence trade from the point of view of a partially informed investor. Therefore we now assume that the investor cannot directly observe the state of the Markov chain $Y$, and that her information comes from the observation of price processes $S^{(m)}, S^{(1)}$ and $ S^{(2)}$. Mathematically, the available information flow is given by filtration $\bF:=\{\F_t,  \ t \in [0,T]\}$, where $\F_t=\sigma(S^{(1)}_u, S^{(2)}_u, S^{(m)}_u, \ 0 \le u \le t)$. Since the investor chooses how to allocate her wealth according to the available information, we will now consider the following set admissible strategy.

\begin{definition}\label{def:strategies_partial_info}
  An $\bF$-admissible portfolio strategy is a  self-financing and $\bF$-predictable strategy  $h=(\hm, \ho, \hd)$ that satisfies integrability condition \eqref{eq:integrability}, and  $\mathcal{A}^{\bF}$ is the set of all $\bF$-admissible strategies.
\end{definition}

In order to solve the optimization problem under partial information we first infer information about the state of the Markov chain $Y$ from the observation process $(S^{(1)},S^{(2)},S^{(m)})$, using filtering techniques. The idea is to determine the conditional distribution of the unobservable state process $Y$, given the observed history. To this, for every functions  $f:E \to \R$ we define the filter $\pi(f)$ as the optional projection of the process $f(Y)$ on the available filtration, i.e.
\begin{align*}
\pi_t(f)=\mathbf{E}\left[f(Y_t)|\F_t\right], \quad t \in [0,T].
\end{align*}
Due to the nature of process $Y$, we get that
\begin{align*}
\pi_t(f)=\sum_{i=1}^K f(e_i)\P\left(Y_t=e_i |\F_t\right) , \quad t \in [0,T].
\end{align*}
Therefore, solving the filtering problem amounts to compute conditional state probabilities,
\begin{align*}
\pi^{i}_t:=\P\left(Y_t=e_i |\F_t\right), \quad t \in [0,T],
\end{align*}
for every $i\in\{1,\dots,K\}$. In the sequel we will use the notation $\bpi$ to indicate the $K$-dimensional process $(\pi^{1}, \dots, \pi^{K} )^\top$ and $\pi_t(f)=\mathbf{f}^\top \bpi_t=\sum_{i=1}^K f^i\pi^i_t$ where ${\bf f}=(f^1, \dots, f^K)^\top$ and $f^i=f(e_i)$ for every $i\in\{1, \dots, K\}$.
We will use the {\em innovations approach} to characterize processes $\pi^i$ for $i \in \{1, \dots, K\}$. Consider the following processes
\begin{align*}
R^{(1)}_t=-\int_0^t\lambda_1(Y_s) (X_s-\alpha_1(Y_s)) \ud s +\sigma B^{(0)}_t + b_1 B^{(1)}_t, \quad t \in [0,T],\\
R^{(2)}_t=\int_0^t\lambda_2(Y_s)( X_s-\alpha_2(Y_s)) \ud s +\sigma B^{(0)}_t + b_2 B^{(2)}_t, \quad t \in [0,T],
\end{align*}
and observe that the triplets $(S^{(1)}, S^{(2)}, S^{(m)})$ and $(R^{(1)}, R^{(2)}, S^{(m)})$ generate the same information flow. Then we define $(\bG, \P)$-Brownian motions $Z^{(1)}$ and $Z^{(2)}$ by
\begin{align*}
Z^{(1)}_t=\frac{\sigma B^{(0)}_t + b_1 B^{(1)}_t}{\sigma_1}, \qquad Z^{(2)}_t=\frac{\sigma B^{(0)}_t + b_2 B^{(2)}_t}{\sigma_2},\quad  t \in [0,T],\end{align*}
where $\sigma_1=\sqrt{\sigma^2+b^2_1}$ and $\sigma_2=\sqrt{\sigma^2+b^2_2}$. Note that $Z^{(1)}$ and $Z^{(2)}$ are correlated Brownian motions with correlation coefficient $\rho=\frac{\sigma^2}{\sigma_1\sigma_2}\in[0,1]$ and that there exists a $(\bG, \P)$-Brownian motion $\widetilde{Z}^2$ independent of $Z^1$ such that
$Z^2=\rho Z^1+\sqrt{1-\rho^2}\widetilde{Z}^2$. We now introduce the innovation process $I=(I^{(1)}, I^{(2)})^\top$ in the following way. Define
\begin{align*}
\mu_1(X_t, Y_t)= -\lambda_1(Y_t) (X_t-\alpha_1(Y_t)), \quad t \in [0,T],\\
\mu_2(X_t, Y_t)=\lambda_2(Y_t) (X_t-\alpha_2(Y_t)), \quad t \in [0,T],
\end{align*}
and denote by $\pi_t(\mu_i)=\mathbf{E}\left[\mu_i(X_t,Y_t)|\F^S_t\right]=\bmu_i(X_t)^\top \bpi_t$ where,  for $i\in\{1,2\}$ and $t \in [0,T]$ vector $\bmu_i(X_t)=(\mu_i(X_t, e_1), \dots,\mu_i(X_t, e_K) )$; then we get that
\begin{align*}
I^{(1)}_t\!\!&=\!Z^{(1)}_t+\int_0^t \frac{\mu_1(X_u, Y_u) - \bmu_1(X_u)^\top \bpi_u }{\sigma_1}  \ud u, \\
I^{(2)}_t\!\!&=\!\widetilde{Z}^{(2)}_t\!+\!\!\int_0^t\!\! \frac{\sigma_1 (\mu_2(X_u, Y_u) - \bmu_2(X_u)^\top \bpi_u)-\rho \sigma_2 (\mu_1(X_u, Y_u) - \bmu_1(X_u)^\top \bpi_u)}{\sigma_1\sigma_2\sqrt{1-\rho^2}} \ud u,
\end{align*}
 for every $t \in [0,T]$. We will use also the matrix/vector form
\begin{align*}
I_t=Z_t+\int_0^t\Sigma^{-1}(A(X_u,Y_u)-\pi_u(A))\ud u, \quad t \in [0,T],
\end{align*}
where $Z=(Z^{(1)}, \widetilde{Z}^{(2)})^\top$, $A(X,Y)=(\mu_1(X,Y), \mu_2(X,Y))^\top$,
\begin{align*}
\Sigma=\left(
\begin{array}{cc}
\sigma_1&0\\
\sigma_2 \rho & \sigma_2 \sqrt{1-\rho^2}
\end{array}\right). \end{align*}

\begin{rem}\label{inn} The innovation process $I$ has two important features. First, $I$ is an  $(\bF, \P)$-Brownian motion; see, for instance, \citet[Proposition 2.30]{bc}. Second, by the independence between the Markov chain $Y$ and the vector $(B^{(m)},B^{(0)}, B^{(1)}, B^{(2)})$ driving the observation we get that the filtration generated by $(S^{(m)}, R^{(1)}, R^{(2)})$ and that generated by  $(S^{(m)}, I^{(1)}, I^{(2)})$ are the same; see \citet[Theorem 1]{allinger1981new}. Then we can apply \citet[Theorem III.4.34-(a)]{JS} and get that every $(\bF, \P)$-local martingale $M$ admits the following representation
\begin{align}\label{eq:mg_representation}
M_t=M_0+\int_0^t\gamma_u\, \ud I_u,\quad t\in[0,T],
\end{align}
for some $\bF$-predictable 2-dimensional process $\gamma$ such that
\begin{equation*}
\int_0^T\|\gamma_u\|^2\,\ud u<\infty, \quad \P-\text{a.s.}
\end{equation*}
\end{rem}
The filtering equation is computed in the next proposition. The proof of this result is given in Appendix.
\begin{proposition}\label{prop:filtering}
For every $i \in \{1, \dots, K\}$, conditional state probabilities of the process $Y$ satisfy the following system of SDEs
\begin{align}\label{eq:conditional_prob}
  \ud \pi_t^i= \sum_{j=1}^K q^{ji}\pi^j_t \ud t +  H^{i}(X_t,\bpi_t) \ud I_t
  \end{align}
with $\pi_0^i=p_0$, where for $i=1, \dots, K$, $H^{i}(X, \bpi):=\{H^{i}(X_t, \bpi_t), t\geq 0\}$ is the 2-dimensional process with components
\begin{align*}
 H^{i,(1)}(X_t, \bpi_t)&= \frac{\pi^i_t (\mu_1(X_t, e_i)-\bmu_1(X_t)^\top \bpi_t)}{\sigma_1} \\
 H^{i,(2)}(X_t, \bpi_t)&=\frac{\pi^i_t\left(\sigma_1(\mu_2(X_t, e_i)-\bmu_2(X_t)^\top\bpi_t)-\sigma_2\rho(\mu_1(X_t, e_i)-\bmu_1(X_t)^\top \bpi_t)\right)}{\sigma_1\sigma_2\sqrt{1-\rho^2}},
\end{align*}
for every $t \in [0,T]$ with $\sigma_1=\sqrt{\sigma^2+b^2_1}$ and $\sigma_2=\sqrt{\sigma^2+b^2_2}$.
\end{proposition}

Having the dynamics of the filtered probabilities enables us to derive a semimartingale decomposition for the co-integrated asset price processes with respect to the information filtration. Precisely, we have that
\begin{align}\label{eq:convergence_pi}
\frac{\ud S^{(1)}_t}{S^{(1)}_t}&=(r+\beta_1 \mu_m)\ud t+\beta_1\sigma_m \ud B^{(m)}_t+\sigma_1 \ud I^{(1)} + \bmu_1(X_t)^\top\bpi_t \ud t,\\
\label{eq:convergence2_pi}
\frac{\ud S^{(2)}_t}{S^{(2)}_t}&=(r+\beta_2 \mu_m)\ud t+\beta_2\sigma_m \ud B^{(m)}_t\!\!+\sigma_2 \rho \ud I^{(1)}_t\!\!+\sigma_2 \sqrt{1-\rho^2}\ud I^{(2)}_t\!\!+\bmu_2(X_t)^\top\!\bpi_t  \ud t,
\end{align}
with $S^{(1)}_0>0$ and $S^{(2)}_0>0$, and the market index price process $S^{(m)}$ which is not affected by the Markov chain, preserves its dynamics. This leads to the following for the spread and the wealth processes
\begin{align}
\nonumber \ud X_t=&\left(\Gamma_1+\left(\bmu_1(X_t)^\top\bpi_t -\bmu_2(X_t)^\top\bpi_t \right) \right)dt\\
&+\left(\beta_1-\beta_2\right)\sigma_m \ud B^{(m)}_t+(\sigma_1 -\rho \sigma_2)\ud I^{(1)}_t-\sigma_2\sqrt{1-\rho^2} \ud I^{(2)}_t, \quad X_0\in \R,\end{align}
\begin{align}
\nonumber \frac{\ud W^h_t}{W^h_t}=&\left(r\!+\!\!\left(\hm_t\!+\ho_t\beta_1\!+\hd_t\beta_2\! \right)\mu_m\!+\!\left(\ho_t \bmu_1(X_t)^\top\!\bpi_t\! +\hd_t \bmu_2(X_t)^\top\!\bpi_t  \right)\!\right)\ud t\\
\nonumber&+\sigma_m\left(\hm_t+\ho_t\beta_1+\hd_t\beta_2\right)\ud B^{(m)}_t+ (\sigma_1 \ho_t +\rho \sigma_2 \hd)\,\ud I^{(1)}_t\\
&+\sigma_2 \sqrt{1-\rho^2} \hd_t \ud I^{(2)}_t, \quad W_0^h>0,
\label{eq:wealthdyn_pi}\end{align}
respectively.
Moreover, thanks to uniqueness of the solution of the filtering equation we can consider the $(K+2)$-dimensional process $(W,X,\pi)$ as the state process and introduce the equivalent optimal control problem under full information, called the {\em separated} problem; see, e.g., \citet{fleming1982optimal}.
The optimization problem we address now is
\begin{align}\label{optim2}
\mbox{  Maximize } \quad \mathbf{E}^{t,w,x, \p}[\log{W_T}] \quad \mbox{ over all } \  h\in\mathcal{A}^\bF\end{align}
where $\mathbb{E}^{t,w,x,\p}$ denotes the conditional expectation given $W_t=w$, $X_t=x$ and $\bpi_t=\p$, where $(w,x,\p)\in\mathbb R_+\times \mathbb R \times \Delta_K$, with $\Delta_K$ denoting the $(K-1)$-dimensional simplex. Next, we resort to the HJB approach to solve  problem \eqref{optim2}.
We define the value function by
\begin{equation}\label{eq:vf_pi}V(t,w,x,\p):=\underset{h\in\mathcal{A}^\bF}\sup\,\,\mathbf{E}^{t,w,x,\p}\left[\log{W_T}\right].\end{equation}

In order to get explicit form for the value function up to the solution of a system of PDEs we restrict to the case where $\lambda_1(y)=\lambda_1\in \R$ and $\lambda_2(y)=\lambda_2\in \R$. In this case coefficients $H^{(i),1}(X, \pi)$ and $H^{(i),2}(X, \pi)$, for $i=1,\dots,K$ in equation \eqref{eq:conditional_prob} do not depend on $X$  and  are given by
\begin{align*}
 H^{i,(1)}(X_t, \bpi_t)&= H^{i, (1)}(\bpi_t)=\frac{\lambda_1 \pi^i_t (\alpha_1^i-\balpha_1^\top \bpi_t)}{\sigma_1}, \\
 H^{i,(2)}(X_t, \bpi_t)&= H^{i, (2)}(\bpi_t)=\frac{-\lambda_2\sigma_1\pi^i_t(\alpha_2^i-\balpha_2^\top\bpi_t)-\sigma_2\lambda_1\rho\pi^i_t(\alpha_1^i-\balpha_1^\top \bpi_t)}{\sigma_1\sigma_2\sqrt{1-\rho^2}},
\end{align*}
for every $t \in [0,T]$.

\begin{theorem}\label{thm:optimal_partial} Suppose that $\lambda_1(y)=\lambda_1\in \R$ and $\lambda_2(y)=\lambda_2\in \R$ with $\lambda_1+\lambda_2>0$ and assume that the investor has logarithmic utility preferences. Then the optimal portfolio strategy $h^*=({\ho}^{\ast},{\hd}^{\ast},{\hm}^{\ast})\in \mathcal{A}^{\bF}$ is
\begin{eqnarray}\label{opt1_pi}
{\ho}^{\ast}(t,x,\p)&=& \frac{\bmu_1(x)^\top \p-\bmu_2(x)^\top\p\varrho_2}{b_1^2+b_2^2\varrho_2}, \\
\label{opt2_pi}
{\hd}^{\ast}(t,x,\p)&=&\frac{\bmu_2(x)^\top \p-\bmu_1(x)^\top \p\varrho_1}{b_2^2+b_1^2\varrho_1},\\
\label{opt3_pi}{\hm}^{\ast}(t,x,\p)&=& \frac{\mu_m}{\sigma_m^2}-\beta_1{\ho}^{\ast}(t,x,\p)-\beta_2{\hd}^{\ast}(t,x,\p).
\end{eqnarray}
The value function is of the form \begin{equation}V(t, w, x, \p )=\log(w)+\bar m(t)x^2+\bar n(t,\p)x+\bar u(t,\p),\end{equation}
where function $\bar m(t)$ solves the ordinary differential equation
\begin{align}\label{eq:system1_pi}
&\bar m_t(t) -2  \bar m(t) (\lambda_1+\lambda_2)+ \Theta_1 =0
\end{align}
with terminal condition $\bar m(T)=0$ and functions $\bar n(t,\p)$ and $\bar u(t,\p)$  solve the following system of partial differential equations
\begin{align}
\nonumber&\bar n_t(t,\p) - \bar n(t, \p) (\lambda_1+\lambda_2)+ \sum_{i,j=1}^K \bar n_{p^i}(t,\p) q^{ji}p^j + 2(\Gamma_1+\lambda_1\balpha_1^\top \p+\lambda_2\balpha_2^\top \p) \bar m(t)\\
&+\frac{1}{2}\sum_{i,j=1}^K \bar n_{p^ip^j}(t,\p)\left(H^{i, (1)}(\p) H^{j, (1)}(\p)+H^{i, (2)}(\p) H^{j, (2)}(\p)\right)- \Theta_2(\p)=0,\label{eq:system2_pi}
\end{align}

\begin{align}
\nonumber&\bar u_t(t,\p)+ \Gamma_2 \bar m(t) + (\Gamma_1+\lambda_1\balpha_1^\top \p+\lambda_2\balpha_2^\top \p) \bar n(t,\p)+  \Theta_3(\p)\\
&\nonumber +\sum_{i=1}^K \bar u_{p^i}(t,\p) q^{ji} p^j+\frac{1}{2}\sum_{i,j=1}^K\bar u_{p^ip^j}(t,\p)(H^{i, (1)}(\p) H^{j, (1)}(\p)+H^{i, (2)}(\p) H^{j, (2)}(\p)) \\
&+  \sum_{i=1}^K \bar n_{p^i}(t,\p) \left(\frac{b_1^2}{\sqrt{\sigma^2+b_1^2}}H^{i, (1)}(\p) -  \frac{\sqrt{\sigma^2(b_1^2+b_2^2)+b_1^2b_2^2}}{\sqrt{\sigma^2+b_1^2}}   H^{i, (2)}(\p)\right) =0
\label{eq:system3_pi}
\end{align}
with terminal conditions $\bar n(T,\p)=0$ and $\bar u(T,\p)=0$ and where $\Gamma_1$ and $\Gamma_2$ are the same of Theorem \ref{optimalfull} and
\begin{align*}
\Theta_1&=\frac{(\sigma^2+b_2^2)\lambda_1^2+(\sigma^2+b_1^2)\lambda_2^2+2\sigma^2 \lambda_1\lambda_2}{2(\sigma^2b_1^2+\sigma^2b_2^2+b_1^2b_2^2)},\\
\Theta_2(\p)&= \frac{\lambda_1 \balpha_1^\top \p(\lambda_1(b_2^2+\sigma^2)+\lambda_2\sigma^2)+\lambda_2\balpha_2^\top \p(\lambda_2(b_1^2+\sigma^2)+\lambda_1\sigma^2)}{b_1^2b_2^2 +\sigma^2(b_1^2+b_2^2)},\\
\Theta_3(\p)&=\frac{(\lambda_1 b_2\balpha_1^\top \p)^2+(\lambda_2 b_1\balpha_2^\top \p)^2+\sigma^2(\lambda_1 \alpha_1^\top \p +\lambda_2\alpha_2^\top \p)^2}{2\left(b_1^2b_2^2 +\sigma^2(b_1^2+b_2^2)\right)}+r+\frac{\mu_m^2}{2\sigma_m^2}.\end{align*}
\end{theorem}
The proof of Theorem \ref{thm:optimal_partial} is given in Appendix. We observe here that the function $\bar m$ driving the quadratic term is independent of $\p$. Mathematically this is due to the fact that $\lambda_1$ and $\lambda_2$ are assumed to be constant and therefore the
trader does not account for the effect of partial information on the quadratic level of the current spread. Optimal portfolio strategy under partial information shares similar properties of the full information one (see Remark \ref{remarkfull}), except that unobserved parameters are replaced by the filtered estimates. That is, the \emph{certainty equivalence principle} holds for the optimization problem under partial information; see \citet{kuwana1995certainty} and \citet{bauerle2004port}.

\subsection{Optimal Beta-Neutral Investment under Partial Information}
For comparison purposes we also investigate the structure of  strategies leading to zero (CAPM) beta in the partial information setting. This means to consider investment strategies of the form outlined below.
\begin{definition}An $\bF$-{\em admissible beta-neutral investment strategy} is an  $\bF$-predictable  self-financing investment strategy $h^{\beta}=({h^{(\beta,1)}},{h^{(\beta,2)}},{h^{(\beta,m)}})$ such that
\begin{align}\label{beta-admisspart}
\beta_1h^{(\beta,1)}_t+\beta_2h^{(\beta,2)}_t=0, \quad t \in [0,T]
\end{align}
with $\esp{\int_0^T\!\!\left(\! {h^{(\beta,m)}_t}^2\!\!\!+{h^{(\beta,1)}_t}^2\right)\ud t}<\infty$. We denote by $\mathcal{A}^{\bF,\beta}$ the set of all $\bF$-admissible beta-neutral strategies.
\end{definition}
The optimal beta-neutral investment strategy under restricted information and the corresponding value function are given in Theorem \ref{thm:optimal_partialB} below. The proof is similar to that of Theorem \ref{thm:optimal_partial} and it is therefore omitted.
\begin{theorem}\label{thm:optimal_partialB}Assume that  $\lambda_1(y)=\lambda_1\in \R$ and $\lambda_2(y)=\lambda_2\in \R$ with $\lambda_1+\lambda_2>0$ and consider a trader with a logarithmic utility function. Then, the optimal beta-neutral strategy under partial information ${h^{\beta}}^{\ast}=({h^{(\beta,1)}}^{\ast},{h^{(\beta,2)}}^{\ast},{h^{(\beta,m)}}^{\ast}) \in \mathcal{A}^{ \bF,\beta}$ is
\begin{align*}
&{h^{(\beta,1)}}^{\ast}(t,x,\p)= -\frac{\lambda_1(x-\balpha_1^\top \p)+\frac{\beta^1}{\beta^2}\lambda_2(x-\balpha_2^\top \p)}{b_1^2+b_2^2\frac{\beta_1^2}{\beta^2_2}+\sigma^2\left(1-\frac{\beta_1}{\beta_2}\right)^2}, \\
&{h^{(\beta,2)}}^{\ast}(t,x,\p)=-\frac{\beta^1}{\beta^2}{h^{(\beta,1)}}^{\ast}(t,x,i),\\
&{h^{(\beta,m)}}^{\ast}(t,x,\p)=\frac{\mu_m}{\sigma_m^2}.
\end{align*}
The value function is of the form \begin{equation*}V(t, w, x, \p )=\log(w)+\bar m(t)x^2+\bar n(t,\p)x+\bar u(t,\p),\end{equation*}
where function $\bar m(t)$ solves the ordinary differential equation
\begin{align*}%\label{eq:system1_piB}
&\bar m_t(t) -2\bar m(t) (\lambda_1+\lambda_2) + \Phi_1=0,
\end{align*}
with the terminal condition $\bar m(T)=0$ and functions $\bar n(t,\p)$ and $\bar u(t,\p)$  solve the following system of partial differential equations
\begin{align*}
\nonumber&\bar n_t(t,\p) +2  \bar m(t)(\Gamma_1+\lambda_1\balpha_1+\lambda_2\balpha_2)- \bar n(t, \p) (\lambda_1+\lambda_2)+ \sum_{i,j=1}^K \bar n_{p^i}(t,\p) q^{ji}p^j \\
&+\frac{1}{2}\sum_{i,j=1}^K \bar n_{p^ip^j}(t,\p)\left(H^{i, (1)}(\p) H^{j, (1)}(\p)+H^{i, (2)}(\p) H^{j,(2)}(\p)\right)- \Phi_2(\p)=0,%\label{eq:system2_piB}
\end{align*}
\begin{align*}
\nonumber&\bar u_t(t,\p)+ (\Gamma_1+\lambda_1\balpha_1+\lambda_2\balpha_2) \bar n(t,\p) +\Gamma_2 \bar m(t) + \Phi_3(\p)\\
&\nonumber +\sum_{i,j=1}^K \bar u_{p^i}(t,\p) q^{ji} p^j+\frac{1}{2}\sum_{i,j=1}^K\bar u_{p^ip^j}(t,\p)(H^{i, (1)}(\p) H^{j, (1)}(\p)+H^{i, (2)}(\p) H^{j,(2)}(\p)) \\
&+  \sum_{i=1}^K \bar n_{p^i}(t,\p) \left(\frac{b_1^2}{\sqrt{\sigma^2+b_1^2}}H^{i, (1)}(\p) -  \frac{\sqrt{\sigma^2(b_1^2+b_2^2)+b_1^2b_2^2}}{\sqrt{\sigma^2+b_1^2}}   H^{i, (2)}(\p)\right) =0
\end{align*}
with terminal conditions $\bar n(T,\p)=0$ and $\bar u(T,\p)=0$ and where $\Gamma_1$ and $\Gamma_2$ are the same of Theorem \ref{optimalfull} and
\begin{align*}
\Phi_1&=\frac{\left(\beta_2\lambda_1+\beta_1\lambda_2\right)^2}{2\left(b_1^2\beta_2^2+b_2^2\beta_1^2+\sigma^2\left(\beta_1-\beta_2\right)^2\right)},\\
\Phi_2(\p)&=\frac{\left(\beta_2\lambda_1 \balpha_1^\top\p +\beta_1\lambda_2\balpha_2^top\p\right)\left(\beta_1\lambda_2+\beta_2\lambda_1\right)}{b_1^2\beta_2^2+b_2^2\beta_1^2+\sigma^2\left(\beta_1-\beta_2\right)^2},\\
\Phi_3(\p)&=\frac{\left(\beta_2\lambda_1\balpha_1^\top \p+\beta_1\lambda_2\balpha_2^\top\p\right)^2}{2\left(b_1^2\beta_2^2+b_2^2\beta_1^2+\sigma^2\left(\beta_1-\beta_2\right)^2\right)} +r+\frac{\mu_m^2}{2\sigma_m^2}.
\end{align*}

\end{theorem}

Finally we again observe that for $\beta_1=\beta_2$ we recover delta-neutral strategies in partial information which are given by
\begin{gather*}
{h^{(\delta,1)}}^{\ast}(t,x,\p)=-{h^{(\delta,2)}}^{\ast}(t,x,\p)=-\frac{\lambda_1(x-\balpha_1^\top \p)+\lambda_1(x-\balpha_1^\top \p)}{b_1^2+b_2^2},\\
\quad {h^{(\delta,m)}}^{\ast}(t,x,\p)=\frac{\mu_m}{\sigma_m^2}.
\end{gather*}

\section{Numerical Study with a 2-State Markov Chain}\label{sec:2stateMC}In this section, we consider a 2-state Markov chain $Y$, that is, $\mathcal{E}=\{e_1, e_2\}$. Here we resort to a numerical approach in order to get qualitative characteristics of optimal strategies and the value function both under full and partial information. In the sequel, we fix the values for the following parameters as $w=1$, $r=0.02$, $\beta_1=1.2$, $\beta_2=1.05$, $\sigma_m=.35$, $\mu_m=0.05$ and $\sigma=0.3$.
%For the rest of parameters, we use different values for the sake of covering interesting cases. We provide results regarding the full and partial information cases.

\subsection{Optimization problem under full information}
We first consider the full information setting where the trader is assumed to observe the state of the Markov chain. We begin with a simulation of a data set where we generate a Markov chain and the price processes, respectively. In Figure~\ref{fig:strategies} we investigate the behavior of the optimal investment strategy for the simulated data. Clearly, we see that strategies depend on different regimes and present jumps at the jump times of the Markov chain. We also observe that the resulting optimal portfolio weights for the first and second assets change sign through time. In particular, we have long-long, long-short and short-short type optimal portfolios, which may indicate the flexibility of our modeling framework.

\begin{figure}[htbp]
\begin{center}
\includegraphics[width=12cm,height=7cm]{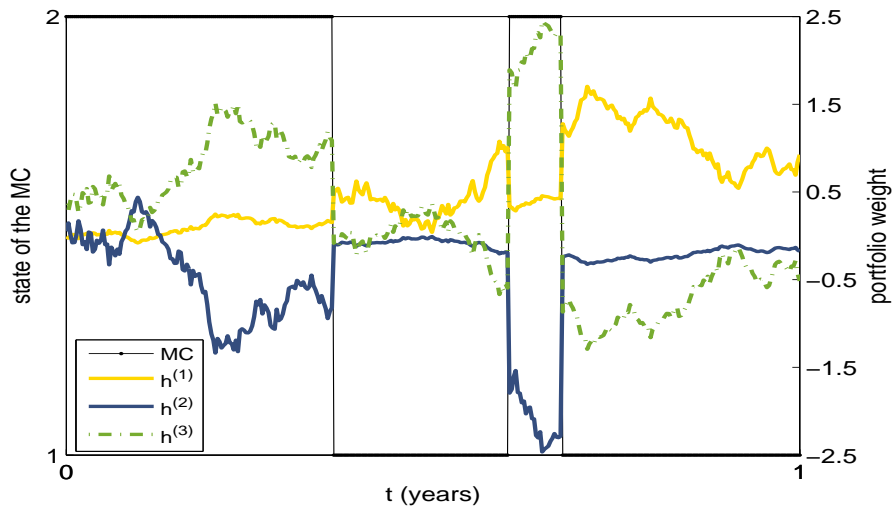}
\caption{Optimal trading strategy for simulated data. Parameter values: $b_1=0.3$, $b_2=0.2$, $\lambda_1^1=0.5$, $\lambda_1^2=-0.3$, $\lambda_2^1=-0.2$, $\lambda_2^2=0.6$, $\alpha_1=\alpha_2=0$, $x=0.01$, $q^{12}=0.01$, $q^{21}=0.02$.   }
\label{fig:strategies}
\end{center}
\end{figure}

Next we investigate the properties of the value function. To this, we solve the system of ODEs in Theorem \ref{optimalfull} numerically. Figure~\ref{fig:value} summarizes our results. Let $(\overline{p},1-\overline{p})$ denote the stationary distribution of the Markov chain $Y$. We consider two traders, one of which ignores the Markov modulated nature of the underlying parameters and use the {\em averaged data} $\overline{\lambda}_i=\overline{p}\lambda_i^1+(1-\overline{p})\lambda_i^2$. The second trader, on the other hand, behaves optimally under our Markov modulated model. We set $q^{12}=0.7$ and $q^{21}=0.2$, and compute $\overline{p}={q^{21}}/(q^{12}+q^{21})=0.22$. Then, we get $\overline{\lambda}_1=-0.12$ and $\overline{\lambda}_2=0.45$. In Figure~\ref{fig:value} we plot $V^{Av}(t,x)$, the value function obtained in the model assuming averaged data, and $\mathbb{E}^{\overline{p}}[V(t,x,Y_t)]=\overline{p} V(t,x,1)+(1-\overline{p})V(t,x,2)$. We observe that $\mathbb{E}^{\overline{p}}[V(t,x,Y_t)]>V^{Av}(t,x)$, that is, averaged data does not suffice to obtain the optimal value for the convergence trade problem and hence on the average, the second trader performs better than the first one. We repeat this analysis for the case of beta-neutral trading and obtain same qualitative results.

In Figure~\ref{fig:value} we also illustrate the dominance of the unrestricted strategies over the beta-neutral ones. This is quite natural since by restricting the set of admissible strategies, the trader could not realize all the benefits resulting from the co-integration between $S^{(1)}$ and $S^{(2)}$. The gap between values depends on the choice of parameters and in particular it increases in initial spread ($x$) and time to maturity ($T-t$).

\begin{figure}[htbp]
\begin{center}
\includegraphics[width=11cm,height=6cm]{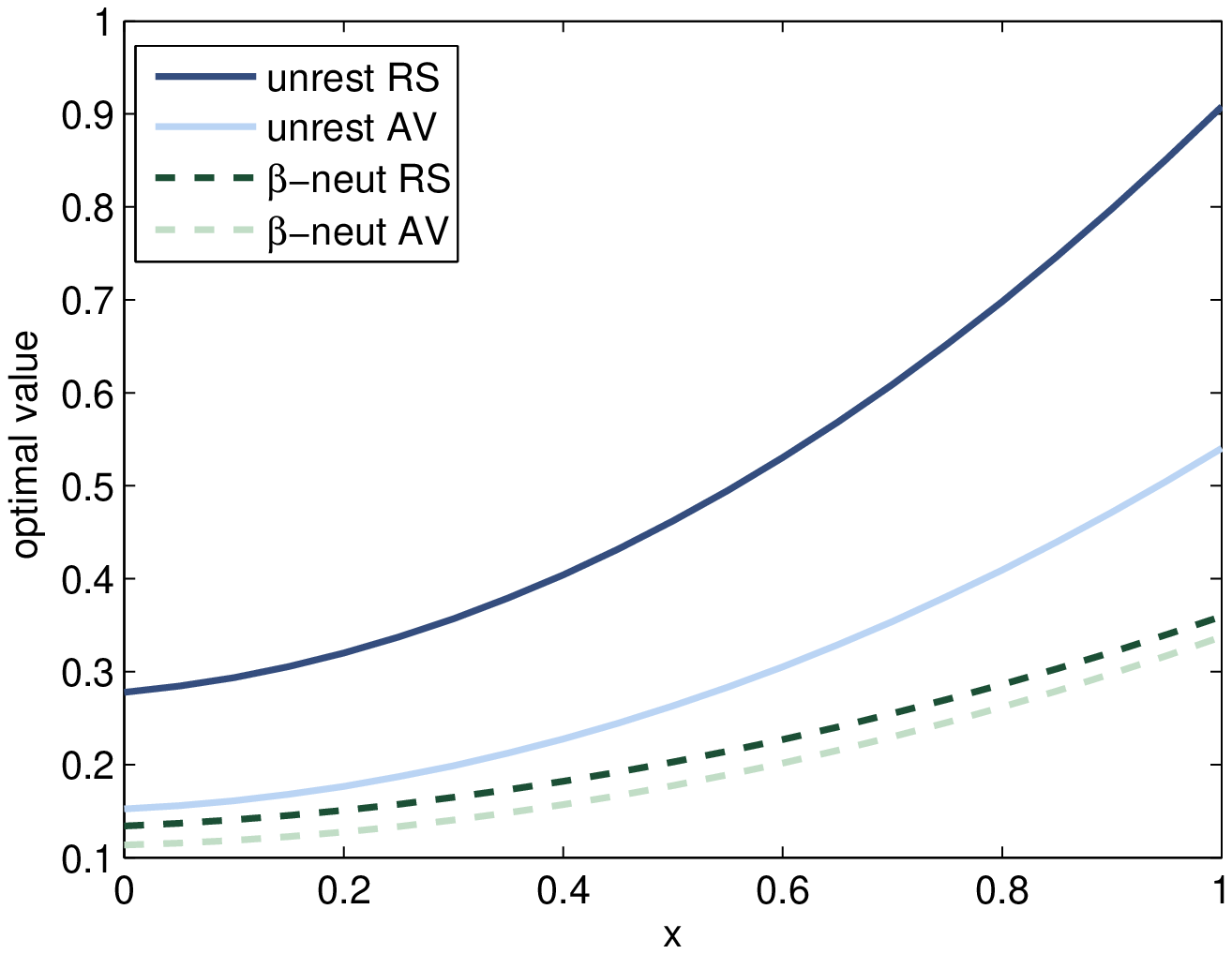}
\includegraphics[width=11cm,height=6cm]{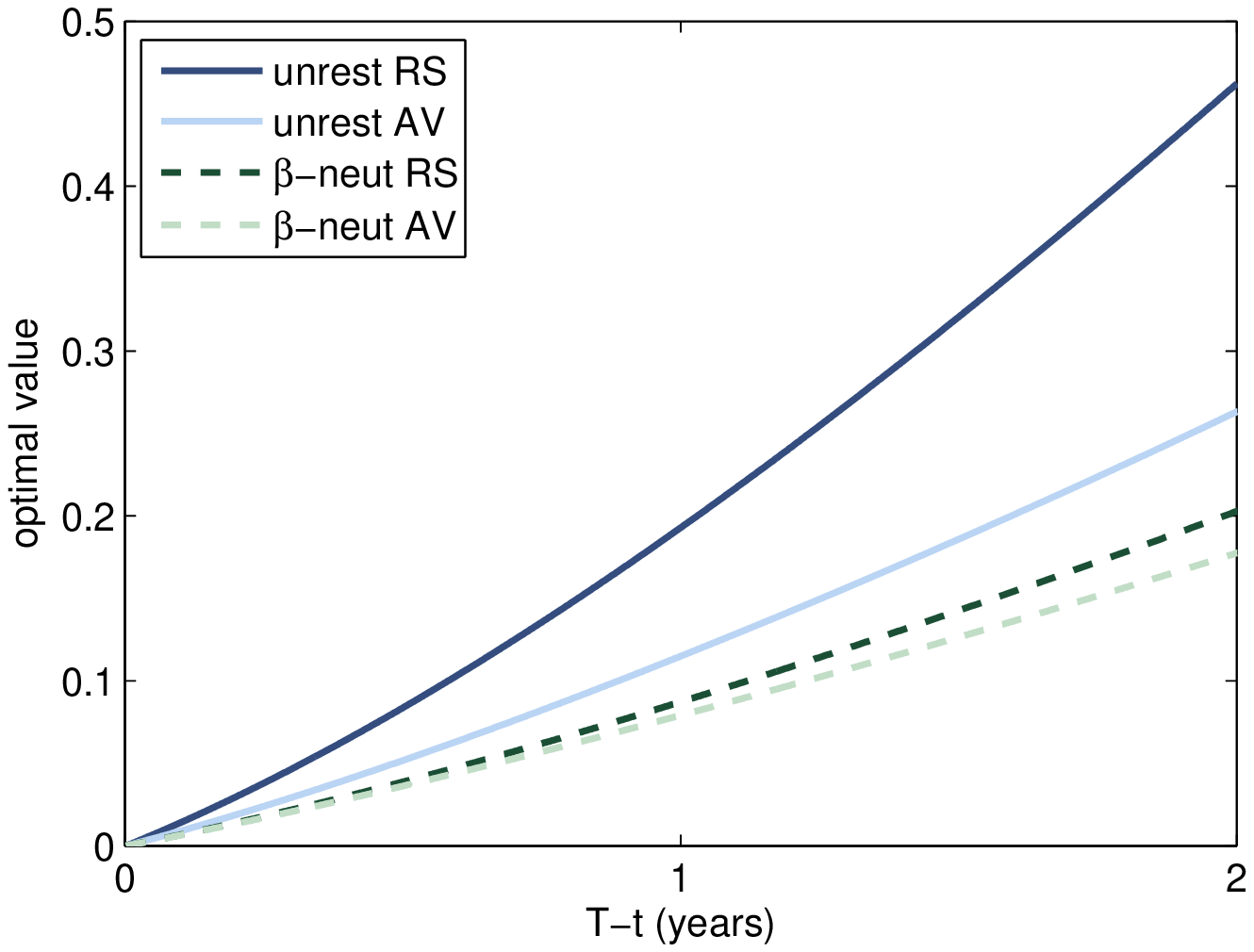}
\caption{Optimal value corresponding to Markov regime-switching (RS) case and averaged data (AV) case for unrestricted and $\beta$-neutral trading as a function of initial spread ($x$) (upper panel) and time to maturity (T-t) (lower panel). Parameter values: T-t=2, $b_1=0.3$, $b_2=0.5$, $\lambda_1^1=0.5$, $\lambda_1^2=-0.3$, $\lambda_2^1=-0.1$, $\lambda_2^2=0.6$, $\alpha_1=\alpha_2= 0$, $x=0.5$, $q^{12}=0.7$, $q^{21}=0.2$.  }
\label{fig:value}
\end{center}
\end{figure}

\subsection{Optimization problem under partial information}
We now consider the partial information case. Since conditional state probabilities $\pi^1$ and $\pi^2$ satisfy $\pi^1_t+\pi^2_t=1$ for every $t \in [0,T]$, we can reduce the number of state variables for the optimization problem. In the following we denote by $p=p^1$ and $1-p=p^2$. In Figure \ref{fig:full_part_strategies} we plot the optimal strategies followed by a fully informed investor who observes  the state of the underlying Markov chain $Y$ and the partially informed one who can only estimate  the state of $Y$ through observation of prices, for the simulated data. We observe that while the informed investor suddenly changes his behavior according to the change in the state of the Markov chain, portfolio weights for the partially informed investor are subject to the smoothing effect of filtering. The difference between portfolio allocation under full and partial information crucially depends on the amplitude of the noise.

\begin{figure}[htbp]
\begin{center}
\includegraphics[width=12cm,height=8cm]{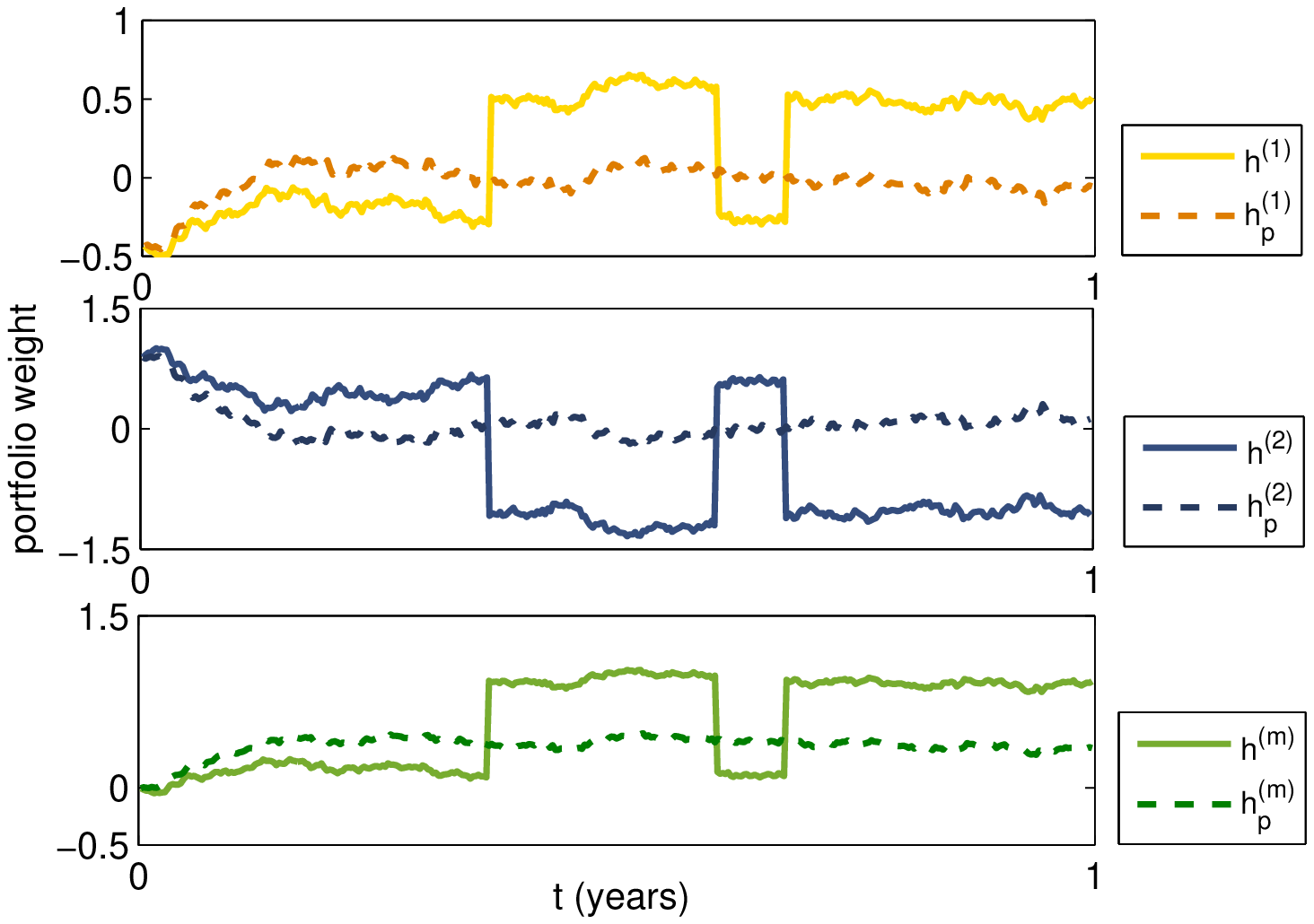}
\caption{Comparison of optimal strategy under full ($h^{(\cdot)}$) and partial information ($h^{(\cdot)}_p$) for simulated data. Parameter values: $\lambda_1\equiv 0.2$, $\lambda_2\equiv 0.15$, $\alpha_1^1=-0.4$,  $\alpha_1^2=0.1$, $\alpha_2^1=0.5$, $\alpha_2^2=-0.5$, $p_0=0$, $x=0.01$, $q^{12}=0.01$, $q^{21}=0.02$. }
\label{fig:full_part_strategies}
\end{center}
\end{figure}

Now we measure the advantage of the fully informed investor over the partially informed one. In order to do that, we consider the process $L$ defined by
\begin{equation*}
L_t=\mathbb{E}\left[V^f(t,W_t,X_t,Y_t)-V^p(t,W_t,X_t, \bpi_t)|\{W_t=w\}\vee \F_t\right], \ t\in [0,T] ,
\end{equation*}
where $V^f$ represents the value function corresponding to the full information setting and $V^p$ that corresponding to the partial information one. The process $L$ represents the loss of utility due to partial information (see, e.g., \citet{lee2016pairs} for the definition). By the form of value functions and Markov property of the pair $(X,\bpi)$ we get that there exists a function $l(t, x, \p)$ such that $L_t=l(t, X_t, \bpi_t)$, $\P-a.s.$ for every $t\in[0,T]$. In Figure \ref{fig:loss} we plot the loss of utility in the 2-state Markov chain case.
We observe that it is always grater than or equal to zero, meaning that information premium exists and it is always nonnegative. Moreover it is larger when conditional state probabilities are close to $0.5$. This reflects the fact that more uncertainty about the state of the Markov chain leads to higher losses in utility.

\begin{figure}[htbp]
\begin{center}
\includegraphics[width=12cm,height=8cm]{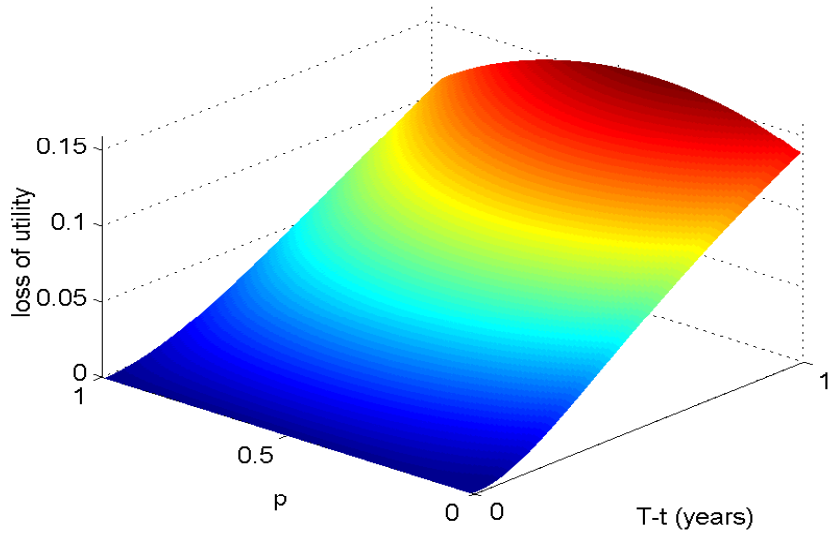}
\caption{Loss of utility due to partial information as a function of estimated state probability (p) and time to maturity (T-t). Parameter values: $\lambda_1\equiv 0.3$, $\lambda_2\equiv 0.4$, $\alpha_1^1=0.5$, $\alpha_1^2=-0.2$, $\alpha_2^1=0.2$, $\alpha_2^2=-0.3$, $x=0.05$, $q^{12}=0.2$, $q^{21}=0.5$. }
\label{fig:loss}
\end{center}
\end{figure}
\section{Conclusion}\label{sec:concl}
In this paper, we have considered an extension to a regime-switching framework of the model proposed by \citet{liu2013optimal}. We have studied the optimization problem for a trader with logarithmic utility preferences under different levels of information. We have assumed that the mean-reverting component of pricing errors depends on a hidden Markov switching factor which may or may not be directly observed by the investor.

In the full information setting, that is when the state of the Markov chain is observable, we have computed the optimal strategy and characterized the value function as the unique (classical) solution of the HJB equation. In this framework, we can reduce the HJB to a system of ODEs. We have also discussed the structure of beta-neutral strategies, achieved by taking long and short positions in such a way that the impact of the overall market is minimized. In the partial information case, we have transformed the original problem into the so-called reduced (or separated) problem via filtering by replacing unobservable states of the Markov chain with their optional projections over the available filtration. Then we have addressed the resulting control problem by dynamic programming, and we have represented the value function in terms of the solution of a system of PDEs. Beta-neutral strategies are also obtained in the partial information framework. Finally, we have studied a numerical example with a two-state Markov chain. We have concluded that averaged data is not sufficient to obtain the optimal value in the full information case, and that there is always positive premium due to information superiority when we compare the optimal value under full and partial information.

\setcounter{equation}{0}
\renewcommand\theequation{A.\arabic{equation}}
\section*{Appendix: Proofs}
\begin{proof}[Proof of Theorem \ref{optimalfull}]
Existence: We denote by $\mathcal L^h_\bG$ the generator of the process $(t,W,X,Y)$, that is
\begin{align*}
\nonumber&\mathcal L^h_\bG F(t,w,x,i)=F_t(t,w,x,i)\!+\!\left(\!\Gamma_1 +\lambda_1^i\alpha_1^i+\lambda_2^i\alpha_2^i-\!(\lambda_1^i\!+\!\lambda_2^i)x\! \right)\! F_x(t,w,x,i)\\
\nonumber+&w\!\left(\!r+\!\mu_m\!\left(\! \hm\!+\!\ho\!\beta_1\!+\!\hd\!\beta_2\!\right)\!+\!\hd \!\lambda_2^i(x-\alpha_2^i)\!-\!\ho\! \lambda_1^i(x-\alpha_1^i)\! \right)F_{w}(t,w,x,i) \\
 \nonumber +&\frac{1}{2}\!w^2\! \!\left(\!\sigma_m^2\!\left(\hm\!\!+\!\ho\!\beta_1\!+\!\hd\!\beta_2\!\right)\!\!^2\!+\!\sigma^2\!\!\left(\!\ho\!\!+\!\hd\!\right)\!^2\!+\!(\ho\!b_1)^2\!+\!(\hd\!b_2)^2\!\right)\!\!F_{ww}(t,w,x,i)\\
 \nonumber+&w\!\left(\!\sigma_m^2\!(\!\beta_1\!-\!\beta_2\!)\!\left(\!\hm\!+\!\ho\!\beta_1\!+\!\hd\!\beta_2\!\right)\!\!+\!\ho\! b_1\!-\!\hd\! b_2 \right)\!F_{wx}(t,w,x,i)\!\!\\
+&\!\frac{1}{2}\Gamma_2F_{xx}(t,w,x,i)+\!\!\sum_{j=1}^K\!\! F\!(t,w,x,j)q^{ij}\!\!,
\end{align*}
for every function $F(\cdot, i)\in C^{1,2,2}([0,T]\times\mathbf{R}_{+}\times\mathbf{R})$, i.e. bounded, differentiable with respect to $t$  and twice differentiable with respect to $w$ and $x$,  for every $i\in \{1, \dots, K\}$.

Suppose that the value function $V(\cdot, i)\in C^{1,2,2}([0,T]\times\mathbf{R}_{+}\times\mathbf{R})$  for every $i\in\{1,\dots,K\}$. Then it solves the HJB equation given by
\begin{equation}\label{pvuci}
 0=\underset{h\in\mathcal{A}}\sup \mathcal L^h V(t,w,x,i)
\end{equation}
for every $i\in \{1,\dots,K \}$, subject to the terminal condition $V(T,w,x,i)=\log(w)$, for all $(w,x)\in\mathbf{R}_{+}\times\mathbf{R}$ and $i\in \{1,\dots,K \}$.
It follows from the form of the utility function that for all $i\in \{1,\dots,K \}$ the value function can be rewritten as $V(t,w,x,i)=\log(w)+\nu(t,x,i)$, for some function $\nu(t,x,i)$ such that $\nu(T,x,i)=0$.
Inserting the ansatz for the value function in equation \eqref{pvuci} and taking first order conditions leads to
\begin{align*}
 0=&\!\frac{\mu_m}{\sigma_m^2}-\ho\beta_1-\hd\beta_2-\hm,\\
 0=&\!\beta_1\mu_m-\!\lambda_1^i\!(x\!-\alpha_1^i\!)-\!\beta_1\sigma_m^2\!\left(\!\hm\!+\!\ho\!\beta_1\!+\!\hd\!\beta_2\!\right)\!-\!\sigma^2\!(\!\ho\!+\!\hd)\!-\!\ho\! b_1^2 ,\\
 0=&\!\beta_2\mu_m+\!\lambda_2^i\!(x\!-\alpha_2^i\!)-\!\beta_2\sigma_m^2\!\left(\! \hm\!+\!\ho\beta_1\!+\!\hd\!\beta_2\!\right)\!-\!\sigma^2\!(\!\ho\!+\!\hd)\!-\!\hd \!b_2^2. \qquad {}
\end{align*}
Second order conditions imply that portfolio weights given in \eqref{opt1}-\eqref{opt3} are candidates to be optimal strategies. Next, we insert the optimal portfolio weights in the HJB equation. This yields the following PDE:
\begin{align}
\nonumber 0=&\nu_t(t,x,i)+\Theta_1^ix^2-\Theta_2^ix+\Theta_3^i+\!\sum_{j=1}^K \!\nu(t,x,j)q^{ij}\!+\!  \frac{1}{2}\Gamma_2 \nu_{xx}(t,x,i)\\
&+\left(\Gamma_1+\lambda_1^i\alpha_1^i+\lambda_2^i\alpha_2^i-(\lambda_1^i+\lambda_2^i)x\! \right)\! \nu_x(t,x,i). \label{eq:HJBfull}
\end{align}
We conjecture that $\nu(t,x,i)=m(t,i)x^2+n(t,i)x+u(t,i)$. Substituting this ansatz in \eqref{eq:HJBfull} results in a quadratic equation for $x$. Setting the coefficients of the terms $x^2$, $x$ and the independent term to zero yields that the functions $m$, $n$ and $u$ solve the system of ODEs given in \eqref{eq:system1}-\eqref{eq:system3}; see, e.g., \citet[Theorem 3.9]{teschl2012ordinary}).

Verification: In the sequel we verify martingale conditions that ensure that $V$ in \eqref{eq:value_function} is indeed the value function. To this, let $v(t,w,x,i)$ be a solution of the HJB equation \eqref{pvuci} and $h\in \mathcal A$ and admissible control. By It\^{o}'s formula we get
\begin{align*}
&v(T, W_T^h, X_T, Y_T)=v(t,w,x,i)+\int_t^T \mathcal Lv(r, W^h_r, X_r, Y_r)\, \ud r\\
&+\int_t^T \sigma_m v_w(r, W^h_r, X_r, Y_r) W^h_r\left(\hm_r+\ho_r\beta_1+\hd_r\beta_2\right)\ud B^{(m)}_r\\
&+\int_t^T \sigma_m+ v_x(r, W^h_r, X_r, Y_r)  \left(\beta_1-\beta_2\right) \ud B^{(m)}_r\\
&+\int_t^T v_w(r, W^h_r, X_r, Y_r) W^h_r \sigma\left(\ho_r+\hd_r \right) \ud B^{(0)}_r\\
&+\int_t^T \left(v_w(r, W^h_r, X_r, Y_r) W^h_r b_1\ho_r + v_x(r, W^h_r, X_r, Y_r) b_1\right) \ud B^{(1)}_r\\
&+\int_t^T \left(v_w(r, W^h_r, X_r, Y_r) W^h_r b_2\hd_r- v_x(r, W^h_r, X_r, Y_r) b_2\right) \ud B^{(2)}_r\\
&+\int_t^T\sum_{j=1}^Kv(r, W^h_r, X_r, j)-v(r, W^h_r, X_r, Y_{r^-}) (m-\nu)(\ud r \times \{j\}).
\end{align*}
The last term in the expression above corresponds to the compensated integral with respect to the jump measure of $Y$, that is
\begin{align*}
&\int_t^T\sum_{j=1}^Kv(r, W^h_r, X_r, j)-v(r, W^h_r, X_r, Y_{r^-}) (m-\nu)(\ud r \times \{j\})=\\
&\sum_{t\leq r\leq T}\Delta v(r, W^h_r, X_r, Y_r)-\!\! \int_t^T\!\!\!\sum_{j=1}^K\!\!v(r, W^h_r, X_r, j)-v(r, W^h_r, X_r, Y_{r^-}) q^{Y_{r^-} j}\, \ud r.
\end{align*}
where $\Delta v(t, W^h_t, X_t, Y_t)=v(t, W^h_t, X_t, Y_t)-v(t, W^h_t, X_t, Y_{t^-})$ for every $t \in [0,T]$,
\[
m([0,t]\times\{j\}):=\sum_{n \ge 1} \I_{\{Y_{T_n}=j\}}\I_{\{T_n\leq t\}},\quad j\in\{1,\dots,K\},\,\,t\in[0,T],
\]
is the jump measure of Markov chain $Y$ with the compensator
\[
\nu([0,t]\times\{j\})=\int_0^t\sum_{i \neq j} q^{ij}\I_{\{Y_{r^-}=i\}}\,\ud r,\quad j\in\{1,\dots,K\},\,\,t\in[0,T].
\]
and $\{T_n\}_{n \in \bN}$ is the sequence of jump times of $Y$.
Since $v$ satisfies equation \eqref{pvuci} we get
\begin{align*}
&v(T, W_T^h, X_T, Y_T)\leq \\
&\quad v(t,w,x,i)+\!\! \int_t^T\!\!\!\! \sigma_m v_w(r, W^h_r, X_r, Y_r) W^h_r\left(\hm_r+\ho_r\beta_1+\hd_r\beta_2\right)\ud B^{(m)}_r\\
&+\int_t^T \sigma_m+ v_x(r, W^h_r, X_r, Y_r)  \left(\beta_1-\beta_2\right) \ud B^{(m)}_r\\
&+\int_t^T v_w(r, W^h_r, X_r, Y_r) W^h_r \sigma\left(\ho_r+\hd_r \right) \ud B^{(0)}_r\\
&+\int_t^T \left(v_w(r, W^h_r, X_r, Y_r) W^h_r b_1\ho_r + v_x(r, W^h_r, X_r, Y_r) b_1\right) \ud B^{(1)}_r\\
&+\int_t^T \left(v_w(r, W^h_r, X_r, Y_r) W^h_r b_2\hd_r- v_x(r, W^h_r, X_r, Y_r) b_2\right) \ud B^{(2)}_r\\
&+\int_t^T\sum_{j=1}^Kv(r, W^h_r, X_r, j)-v(r, W^h_r, X_r, Y_{r^-}) (m-\nu)(\ud r \times \{j\}).
\end{align*}
The form of $v$ and integrability condition \eqref{eq:integrability} ensure that integrals with respects to Brownian motions $B^{(m)}, B^{(0)},B^{(1)},B^{(2)} $ and the compensated jump measure $m-\nu$ are true $(\bG, \P)$-martingales. Then, taking expectations we get that
\begin{align*}
V(t,w,x,i)\leq v(t,w,x,i),
\end{align*}
and the equality holds if $h$ is a maximizer of equation \eqref{pvuci}.
\end{proof}

%\begin{remark}[Constants in Theorem \ref{optimalfullbeta}]
%\begin{align*}
%\Phi_1^i&=\frac{\left(\beta_2\lambda_1^i+\beta_1\lambda_2^i\right)^2}{2\left(b_1^2\beta_2^2+b_2^2\beta_1^2+\sigma^2\left(\beta_1-\beta_2\right)^2\right)},\\
%\Phi_2^i&=\frac{\left(\alpha_1^i\beta_2\lambda_1^i+\alpha_2^i\beta_1\lambda_2^i\right)\left(\beta_1\lambda_2^i+\beta_2\lambda_1^i\right)}{b_1^2\beta_2^2+b_2^2\beta_1^2+\sigma^2\left(\beta_1-\beta_2\right)^2},\\
%\Phi_3^i&=\frac{\left(\alpha_1^i\beta_2\lambda_1^i+\alpha_2^i\beta_1\lambda_2^i\right)^2}{2\left(b_1^2\beta_2^2+b_2^2\beta_1^2+\sigma^2\left(\beta_1-\beta_2\right)^2\right)} +r+\frac{\mu_m^2}{2\sigma_m^2}.
%\end{align*}
%\end{remark}

\begin{proof}[Proof of Proposition~\ref{prop:filtering}]
In the following we use the notation $\widehat{g(Y_t)}=\mathbf{E}\left[g(Y_t)|\F_t\right]$, $t\in[0,T]$. Consider the semimartingale decomposition of $f(Y)$ given by
\begin{align*}
f(Y_t)= f(Y_0)+\int_0^t \langle Q\mathbf f, Y_{u^-}\rangle\, \ud u + M^{(1)}_t,\quad t \in [0,T],
\end{align*}
where $M^{(1)}$ is a $(\mathbb G, \P)$-martingale. Now, projecting over $\bF$ leads to
\begin{align*}
\widehat{f(Y_t)}- \widehat{f(Y_0)}-\int_0^t \langle Q\mathbf f, \widehat{Y}_{u^-}\rangle\, \ud u = M^{(2)}_t,\quad t \in [0,T],
\end{align*}
where $M^{(2)}$ is an $(\bF, \P)$-martingale. Using the martingale representation in \eqref{eq:mg_representation} we get
\begin{align*}
\widehat{f(Y_t)}- \widehat{f(Y_0)}-\int_0^t \langle Q\mathbf f, \widehat{Y}_{u^-}\rangle\, \ud u = \int_0^t\gamma_u\, \ud I_u,\quad t \in [0,T].
\end{align*}
Let $m_t=I_t+\int_0^tX_u \Sigma^{-1} \widehat{A(X_u,Y_u)}\,\ud u$, for every $t \in [0,T]$. Computing the product $f(Y)\cdot m$ and projecting on $\bF$, we obtain
\begin{align}\label{eq:filt1}
\widehat{f(Y_t)\cdot m_t}= \int_0^t\!\!\! m_u \langle Q\mathbf f, \widehat{Y}_u\rangle \, \ud u + \int_0^t\!\! X_u \Sigma^{-1}\widehat{f(Y_u) A(X_u,Y_u)}\,\ud u + M^{(3)}_t,\,\,
\quad{}\end{align}
for every $ t \in [0,T]$ and for some $(\bF, \P)$-martingale $M^{(3)}$.
We now compute the product  $\widehat{f(Y)}\cdot m$ as
\begin{align}\label{eq:filt2}
\widehat{f(Y_t)} \cdot m_t=\!\! \int_0^t \!\!\!m_u \langle Q\mathbf f, \widehat{Y}_u\rangle\,\ud u \!+\!\! \int_0^t\!\!\! X_u \Sigma^{-1}\widehat{f(Y_u)} \widehat{A(X_u,Y_u)}\, \ud u \!+\!\! \int_0^t\!\!\! \gamma_u \,\ud u \!+ M^{(4)}_t.
\end{align}
for every $t \in [0,T]$, where $M^{(4)}$ is an $(\bF, \P)$-martingale.
Comparing the finite variation terms in \eqref{eq:filt1} and \eqref{eq:filt2}, we get
\begin{align*}
\gamma^{(1)}_t&=\frac{\widehat{f(Y_t) \mu_1(X_t,Y_t)}-\widehat{f(Y_t)}\widehat{\mu_1(X_t,Y_t)}}{\sigma_1},\\
\gamma^{(2)}_t&=\frac{\sigma_1 (\widehat{f(Y_t)\mu_2(X_t,Y_t)}-\widehat{f(Y_t)}\widehat{\mu_2(X_t,Y_t)})-\sigma_2\rho(\widehat{f(Y_t)\mu_1(X_t,Y_t)}-\widehat{f(Y_t)}\widehat{\mu_1(X_tY_t)})}{\sigma_1\sigma_2\sqrt{1-\rho^2}},
\end{align*}
for every $t \in [0,T]$. By taking $f(Y_t)=\I_{\{Y_t=e_i\}}$, we obtain the result.
Finally, since the drift and diffusion coefficients in \eqref{eq:conditional_prob} are continuous, bounded and locally Lipschitz, we get that $\bpi=(\pi^1, \dots, \pi^K)$ is the unique strong solution of the system \eqref{eq:conditional_prob} .
\end{proof}

\begin{proof}[Proof of Theorem \ref{thm:optimal_partial}]
%Let us define the following constants,
\textit{Existence:}
For notational ease we set $\sigma_1=\sqrt{\sigma^2+b_1^2}$ and $\sigma_2=\sqrt{\sigma^2+b_2^2}$. Assume first that function $V(t,w,x,\p)$ is regular. Then it satisfies the following  HJB equation
\begin{align}
0= \underset{h\in\mathcal{A}^{\bF}}  \sup \mathcal L_\bF^h V(t,w,x,\p)\label{eq:HJB_pi}
\end{align}
subject to the terminal condition $V(T,w,x,\p)=\log (w)$, for all $w>0$, $x\in\mathbb{R}$ and for every $\p\in \Delta_K$, where $\mathcal L_\bF^h$ is given by
\begin{align}
\nonumber&\L_\bF^h f(t,w,x,\p)=\bigg{\{} f_t + f_x\left(\Gamma_1+\bmu_1(x)^\top \p-\bmu_2(x)^\top \p\right)+ \sum_{i,j=1}^K f_{p^i} q^{ji}p^j\\
\nonumber&+  \left(  r+\left(\hm+ \ho\beta_1+\!\hd  \beta_2 \right)\mu_m +\ho \bmu_1(x)^\top \p + \hd \bmu_2(x)^\top \p \right) w f_w(t,w,x,\p) \\
\nonumber& + \frac{1}{2} f_{xx}\Gamma_2 + \frac{1}{2} \sum_{i,j=1}^Kf_{p^ip^j} (H^{(i),1}(\p) H^{(j),1}(\p)+H^{(i),2}(\p) H^{(j),2}(\p))\\
\nonumber& +\frac{1}{2} f_{ww}w^2 \left((\hm+ \ho\beta_1+\!\hd  \beta_2 )^2\sigma_m^2 +(\sigma_1 \ho+\rho\sigma_2\hd)^2 + \sigma_2^2 (1-\rho^2){\hd}^2\right)\\
\nonumber& + f_{wx}w \left( \sigma_m^2(\beta_1-\beta_2)(\hm+ \ho\beta_1+\!\hd  \beta_2 )+\sigma_1^2 \ho- \sigma_2^2 \hd-\rho\sigma_1\sigma_2 (\ho-\hd)\right)\\
\nonumber& + \sum_{i=1}^K f_{wp^i}w  \left(H^{(i),1}(\p) (\sigma_1 \ho+\rho\sigma_2\hd)+ H^{(i),2}(\p)\sqrt{1-\rho^2}\hd \sigma_2\right)\\
\label{eq:HJB_pigenerator} & + \sum_{i=1}^K f_{xp^i} \left( (\sigma_1-\rho \sigma_2)  H^{(i),1}(\p) -  H^{(i),2}(\p)\sigma_2\sqrt{1-\rho^2}\right)  \bigg{\}}
\end{align}
for every function $f:[0,T]\times \R^+ \times \R \times \Delta_K \to \R$, which is bounded, differentiable with respect to time and twice differentiable with respect to $(w,x,\p)$ with bounded derivatives. By the form of the utility function we have that the value function has the form $V(t,w,x,\pi)=\log(w)+v(t,x, \pi)$, for some function $v(t,x,\pi)$, such that $v(T,x,\p)=0$ for all $(x,\p)\in (\R\times \Delta_K)$. By inserting the first ansatz in equation \eqref{eq:HJB_pi} and considering the first order condition we get that  the candidate for an optimal strategy is given by \eqref{opt1_pi}, \eqref{opt2_pi},\eqref{opt3_pi}.
Since $V(t,w,x,\p)$ is concave and increasing in $w$, the second order condition implies that \eqref{opt1_pi},\eqref{opt2_pi} and \eqref{opt3_pi} is the maximizer and the optimal portfolio strategy.
Here, we choose $v$ of the form $v(t,x,\p)=\bar m(\p)x^2+\bar n (t,\p)x+\bar u(t,\p)$. Inserting this ansatz in equation \eqref{eq:HJB_pi}   leads to the system of linear partial differential equations in  \eqref{eq:system1_pi}, \eqref{eq:system2_pi}, \eqref{eq:system3_pi}.

\noindent \textit{Verification:}
To conclude that $V$ is the value function, we show a verification result. Let $\widetilde{V}(t,w,x,\p)$ be a solution of \eqref{eq:HJB_pi} with the boundary condition $\widetilde{V}(T,w,x,\p)=\log(w)$.
Let $h\in \mathcal A^{\bF}$ be an $\bF$-admissible control, let $W^h$ the solution to equation \eqref{eq:wealthdyn_pi}.
By applying It\^{o}'s formula we get
\begin{align}
\nonumber& \widetilde{V}(T, W_T^h, X_T,\bpi_T)=  \widetilde{V}(t,z,s,\p)+\int_t^T  \mathcal L_\bF^h \widetilde{V}(u, W^h_u, X_u, \bpi_u)\, \ud u \\
\nonumber&+ \!\!\int_t^T\!\!\!\!  \left(  \widetilde{V}_w(u, W^h_u, X_u, \bpi_u) W^h_u   (\hm_u\!\!+ \!\ho_u\beta_1\!+\! \hd_u  \beta_2 ) \!+\! \widetilde{V}_x(u, W^h_u, X_u, \bpi_u) (\beta_i\!-\!\beta_2) \right)\!\sigma_m \ud B^{(m)}_u\\
\nonumber& +\!\!\int_t^T\!\!\!\!  \left(\widetilde{V}_w(u, W^h_u, X_u, \bpi_u) W^h_u (\sigma_1\ho+\rho \sigma_2 \hd) +  (\sigma_1-\rho\sigma_2) \widetilde{V}_x(u, W^h_u, X_u, \bpi_u)\right)   \ud I^{(1)}_u \\
\nonumber&  +\!\!\int_t^T\!\!\!\!\sum_{i=1}^K\widetilde{V}_{p^i}(u, W^h_u, X_u, \bpi_u)  \bar H^{(i),1}(\bpi_t)    \ud I^{(1)}_u\\
\nonumber& + \!\!\int_t^T\!\!\!\!  \left(\!\sigma_2 \sqrt{1-\rho^2}( \widetilde{V}_w(u, W^h_u, X_u, \bpi_u) W^h_u  \hd_u \!-\!\widetilde{V}_x(u, W^h_u, X_u, \bpi_u)) \right)  \ud I^{(2)}_u\\
\nonumber&+  \!\!\int_t^T\!\! \sum_{i=1}^K\!\widetilde{V}_{p^i}(u, W^h_u, X_u, \bpi_u)  H^{(i), 2} (\pi_t)   \ud I^{(2)}_u.
\end{align}
By equation \eqref{eq:HJB_pi} we get
\begin{align}
\nonumber& \widetilde{V}(T, W_T^h, X_T,\bpi_T)\leq  \widetilde{V}(t,w,x,\p) \\
\nonumber&+ \!\!\int_t^T\!\!\!\!  \left(  \widetilde{V}_w(u, W^h_u, X_u, \bpi_u) W^h_u   (\hm_u\!\!+ \!\ho_u\beta_1\!+\! \hd_u  \beta_2 ) \!+\! \widetilde{V}_x(u, W^h_u, X_u, \bpi_u) (\beta_i\!-\!\beta_2) \right)\!\sigma_m \ud B^{(m)}_u\\
\nonumber& +\!\!\int_t^T\!\!\!\!  \left(\widetilde{V}_w(u, W^h_u, X_u, \bpi_u) W^h_u (\sigma_1\ho+\rho \sigma_2 \hd) +  (\sigma_1-\rho\sigma_2) \widetilde{V}_x(u, W^h_u, X_u, \bpi_u)\right)   \ud I^{(1)}_u \\
\nonumber&  +\!\!\int_t^T\!\!\!\!\sum_{i=1}^K\widetilde{V}_{p^i}(u, W^h_u, X_u, \bpi_u)  \bar H^{(i),1}(\bpi_t)    \ud I^{(1)}_u\\
\nonumber& + \!\!\int_t^T\!\!\!\!  \left(\!\sigma_2 \sqrt{1-\rho^2}( \widetilde{V}_w(u, W^h_u, X_u, \bpi_u) W^h_u  \hd_u \!-\!\widetilde{V}_x(u, W^h_u, X_u, \bpi_u)) \right)  \ud I^{(2)}_u\\
&+  \!\!\int_t^T\!\! \sum_{i=1}^K\!\widetilde{V}_{p^i}(u, W^h_u, X_u, \bpi_u)  H^{(i), 2} (\bpi_t)   \ud I^{(2)}_u.\label{eq:ver2}
\end{align}
Note that stochastic integrals with respect to $B^{(m)},I^{(1)}$ and $I^{(2)}$ are true martingales. This is a consequence of the fact that function  $\widetilde{V}(t,w,x,\p)=\log(w)+\bar m(t)x^2+\bar n(t,\p)x+\bar u(t, \p)$ solves the HJB equation, that $(\hm, \ho, \hd)$ is an $\bF$-admissible strategy and that functions $\bar m(t),\bar n(t,\p), \bar u(t,\p)$ and their derivatives are bounded over the compact interval $[0,T]\times \Delta_K$.
Then taking the expectation on both sides of inequality \eqref{eq:ver2} implies that $V(t,w,x,\p)\leq \widetilde{V}(t,w,x,\p)$.
Moreover if $({\hm}^*,{\ho}^*,{\hd}^
*) $ is a maximizer of equation \eqref{eq:HJB_pi}, then we obtain the equality $V(t,w,x,\p)= \widetilde{V}(t,w,x,\p)$.
\end{proof}

\bibliographystyle{plainnat}
\bibliography{bibtex_optimal_working}

\end{document}